\documentclass[10pt,journal,compsoc]{IEEEtran}
\usepackage{amsmath, amsthm, amssymb}
\usepackage{algorithm}
\usepackage[noend]{algpseudocode}
\usepackage{romannum}

\usepackage{float}
\usepackage{subfig}
\usepackage{graphicx}
\usepackage{longtable}
\usepackage{multirow}

\usepackage{array}
\newcolumntype{C}[1]{>{\centering\arraybackslash}m{#1}}

\newtheorem{theorem}{\textbf{Theorem}}
\newtheorem*{process*}{\textbf{Process}}

\newtheorem{lemma}{\textbf{Lemma}}

\theoremstyle{definition}

\newtheorem{problem}{\textbf{Problem}}
\newtheorem{definition}{\textbf{Definition}}
\newtheorem{eqs}{\textbf{Equation System}}
\DeclareMathOperator*{\argmin}{arg\,min}

\ifCLASSOPTIONcompsoc
  \usepackage[nocompress]{cite}
\else
  \usepackage{cite}
\fi
\ifCLASSINFOpdf
\else
\fi
\begin{document}
%
\title{Coupon Advertising in Online Social Systems: Algorithms and Sampling Techniques}
%
%
%
%

\author{
	Guangmo (Amo)~Tong,~\IEEEmembership{Student Member,~IEEE,}
	Weili~Wu,~\IEEEmembership{Member,~IEEE,}
	and~Ding-Zhu~Du
	\IEEEcompsocitemizethanks{
		\IEEEcompsocthanksitem G. Tong, W. Wu, and D.-Z. Du are with the Department of Computer Science Erik
		Jonsson School of Engineering and Computer Science The University
		of Texas at Dallas, 800 W. Campbell Road; MS EC31 Richardson, TX
		75080 U.S.A.\protect\\
		E-mail: guangmo.tong@utdallas.edu}
}

%
%

\markboth{Journal of \LaTeX\ Class Files,~Vol.~14, No.~8, August~2015}%
{Shell \MakeLowercase{\textit{et al.}}: Bare Demo of IEEEtran.cls for Computer Society Journals}
%



\IEEEtitleabstractindextext{%
\begin{abstract}
Online social systems have become important platforms for viral marketing where the advertising of products is carried out with the communication of users. After adopting the product, the seed buyers may spread the information to their friends via online messages e.g. posts and tweets. In another issue, electronic coupon system is one of the relevant promotion vehicles that help manufacturers and retailers attract more potential customers. By offering coupons to seed buyers, there is a chance to convince the influential users who are, however, at first not very interested in the product. In this paper, we propose a coupon based online influence model and consider the problem that how to maximize the profit by selecting appropriate seed buyers. The considered problem herein is markedly different from other influence related problems as its objective function is not monotone. We provide an algorithmic analysis and give several algorithms designed with different sampling techniques. In particular, we propose the RA-T and RA-S algorithms which are not only provably effective but also scalable on large datasets. The proposed theoretical results are evaluated by extensive experiments done on large-scale real-world social networks. The analysis of this paper also provides an algorithmic framework for non-monotone submodular maximization problems in social networks.
\end{abstract}

}

\maketitle

\IEEEdisplaynontitleabstractindextext

%
\IEEEpeerreviewmaketitle

\IEEEraisesectionheading{\section{Introduction}\label{sec:introduction}}

\IEEEPARstart{O}nline social networks have become important platforms for viral marketing as they allow fast spread of brand-awareness information. As a typical diffusion process, the influence starts to spread from seed users and then goes round by round from active users to inactive users. The classic influence maximization (IM) problem aims to find the seed users that can maximize the number of influenced nodes. In the seminal work \cite{kempe2003maximizing}, Kempe \textit{et al.} propose the triggering model and formulate the IM problem as a combinatorial optimization problem. A set function $h()$ over a ground set $\mathcal{N}$ is submodular if
\begin{equation}
\label{eq: submodular}
h(\mathcal{N}_1\cup\{v\})-h(\mathcal{N}_1) \geq h(\mathcal{N}_2\cup\{v\})-h(\mathcal{N}_2) 
\end{equation}
for any $\mathcal{N}_1 \subseteq \mathcal{N}_2 \subseteq \mathcal{N}$ and $v \in \mathcal{N} \setminus \mathcal{N}_2$.
It turns out that IM is a size-constraint monotone submodular maximization problem and therefore the natural greedy algorithm yields a good approximation \cite{nemhauser1978analysis}. In the past decade, a huge body of works has been devoted to influence related problems such as rumor blocking, active friending and effector detection, just to name a few. The process of influence diffusion generalizes the spread of many kinds of information such as news, ideas and advertisements. In this paper, we crystallize the influence diffusion process and in particular consider the advertising of product in online social systems.

\textbf{Influenced node Vs. Adopter.} Whereas IM is motivated by viral marketing, maximizing the influence may not result in a high profit because an influenced node is not necessarily to be an adopter of the product. As mentioned in Kalish's work \cite{kalish1985new}, the adoption of a new product is characterized by two steps: awareness and adoption. The traditional influence diffusion models characterize the spread of brand-awareness information but they do not take account of the issue of product adoption. Among the factors that affect the customers' behavior, the price is undoubtedly the most relevant one. In this paper, we consider a simplified adoption process by assuming that, once being aware of the product, a customer becomes an adopter if they believe the product is fairly priced, and only an adopter can bring more adopters by passing on the product information further to their friends\footnote{The analysis in paper can be applied to other models where the adoption process is independent from the influence process.}.

\textbf{Coupon Systems.} A real advertising campaign is often conducted with promotional tools such as coupon, reward, free sample and etc. Coupon has been demonstrated to be an effective promotional method as it creates short-term excitement \cite{ailawadi2014sales} and therefore immediately impacts the customers' purchase behavior. Hence, in online social marketing, coupon enables the seller to convince the influential users who were not able to be adopters without a coupon. In this paper, we consider the marketing process where the seed buyers are offered with coupons. From the perspective of sellers, coupons are not free because they essentially lower the price of the product. Therefore, increasing the number of seed buyers does not necessarily bring more profit, which leads the non-monotonicity of the objective function of the considered problem.

\textbf{This paper.} By integrating coupon system into the general triggering model \cite{kempe2003maximizing}, we in this paper propose the triggering-coupon (T-C) model for product advertising process with coupon in social networks. Based on the T-C model, we consider the problem of selecting seed buyers such that the profit can be maximized, which is termed as profit maximization (PM) problem. As discussed later, the difficulty in solving the PM problem lies in two aspects. First, as a function of the seed users, the profit under the T-C model is submodular but not monotone increasing anymore. Consequently, the well-known greedy algorithm does not work and we have to look for other approximation algorithms\footnote{The PM problem is NP-hard and therefore approximation algorithms are the best possible solutions with performance guarantees.} for unconstrained submodular maximization (USM). With the very recent research progress, USM admits several approximation algorithms where the state-of-the-art one achieves the approximation ratio of $\frac{1}{2}$ \cite{buchbinder2015tight}. However, even with such algorithms, the PM problem remains hard to solve because the function value cannot be efficiently computed. The existing USM algorithms assume that there is an oracle of the objective function while, unfortunately, computing the profit under the T-C model is a \#P-hard problem. Since the exact value of the objective function is hard to compute, one naturally seeks help from its estimates obtained by sampling. With different sampling techniques, we design several algorithms for the PM problem. The most straightforward approach to applying the existing algorithms is to obtain the function value by sampling whenever its oracle is called, which we call the forward sampling framework. Under this approach, the estimates can be obtained by either directly simulating the diffusion process or generating realizations\footnote{As introduced later in Sec. \ref{subsec:realization}, the so-called realization is a derandomization of a triggering model.} according to the definition of triggering model. With such two methods, we give the simulation-based profit maximization (SPM) algorithm and the realization-based profit maximization (RPM) algorithm.  Motivated by the study of IM problem, we further design two algorithms, RA-T algorithm and RA-S algorithm, based on the reverse sampling technique. The reverse sampling is initially proposed by Borgs \textit{et al.} \cite{borgs2014maximizing} for IM problem and it is later improved by Tang \textit{el al.} \cite{tang2015influence} and Nguyen \textit{et al.} \cite{nguyen2016stop}. In this paper, we show how the reverse sampling can be used to solve the PM algorithm. All the algorithms proposed in this paper achieve the approximation ratio of $\frac{1}{2}-\epsilon$ with a high probability. Different from most of the existing works, this paper gives an analysis framework for non-monotone submodular problems in social networks. Extensive experiments have been performed for testing the proposed approaches. In the experiments, besides showing the effectiveness and efficiency of our algorithms, we also provide several interesting observations. 

\textbf{Organization.} The rest of this paper is organized as follows. Sec. \ref{sec: related} reviews the related works. The preliminaries are provided in Sec. \ref{sec: pre}. The analysis of the algorithms designed based on forward sampling and reverse sampling is shown in Secs. \ref{sec:front} and \ref{sec:reverse}, respectively. In Sec. \ref{sec:exp}, we present the experiments together with the discussion of the results. Sec. \ref{sec:conclustion} concludes this paper. 

\section{Related Work}
\label{sec: related}
\textbf{Influence Maximization.} Domingos \textit{et al.} \cite{domingos2001mining} are among the first who study the value of customers in a social network. Kempe \textit{et al.} \cite{kempe2003maximizing} propose the IM problem and build the triggering model for influence diffusion. As two concrete triggering models, independent cascade (IC) model and linear threshold (LT) model have been widely adopted in literature. For the IM problem, Kempe \textit{et al.} \cite{kempe2003maximizing} give a greedy algorithm running with Monte Carlo simulations. Leskovec \textit{et al.} \cite{leskovec2007cost} later improve the greedy algorithm with the method of lazy-forward evaluation. However, Monte Carlo simulation is very time-consuming and therefore those two algorithms are not scalable to large datasets. Borgs \textit{et al.} \cite{borgs2014maximizing} invent the reserve sampling technique which essentially provides a better way to estimate the function value. With the reverse sampling technique, Tang \textit{et al.} (\cite{tang2014influence} and \cite{tang2015influence}) design two efficient algorithms which are improved later by Nguyen \textit{et al.} \cite{nguyen2016stop}. Very recently, Arora \textit{et al.} \cite{arora2017debunking} experimentally evaluate the existing algorithms for IM and make a comprehensive comparison. Besides approximation algorithms, efficient heuristics are also developed for large graphs, namely Chen \textit{et al.} \cite{chen2010scalable} and Cheng \textit{et al.} \cite{cheng2013staticgreedy}. 

\textbf{Profit Maximization.} Since the investigation of IM, many influence related problems in social networks have been studied. In what follows, we briefly review the works that are related to profit maximization. One branch of the existing works focuses on pricing issue of the product, and they either formulate the price setting as a game, such as Arthur \textit{et al.} \cite{arthur2009pricing}, Zhou \textit{et al.} \cite{zhou2016bilevel} and  Lu \textit{et al.} \cite{lu2016pricing}, or study the optimal price setting such as Zhu \textit{et al.} \cite{zhu2013influence}. Yang \textit{et al.} \cite{yang2016continuous} consider the problem of finding the optimal discount such that the adoption of the product can be maximized. While also aims to maximize the profit, our paper considers the problem of selecting appropriate seed users instead of designing price strategies. As the most relevant works,  Lu \textit{et al.} \cite{lu2012profit} and Tang \textit{et al.} \cite{tang2016profit} also consider the problem of maximizing profit by selecting high quality seed users. In particular,  Lu \textit{et al.} in \cite{lu2012profit} design a heuristic and Tang \textit{et al.} in \cite{tang2016profit} provide an algorithmic analysis with the utilization of the techniques of USM. Whereas the algorithms proposed in \cite{tang2016profit} have a strong flavor of approximation bound, they are time consuming as they use Monte Carlo simulations. In this paper, our goal is to design profit maximization algorithms that are not only provably effective but also highly scalable.
\section{Preliminaries}
This section provides the preliminaries to the rest of this paper.  
\label{sec: pre}
\subsection{System Model}
\label{subsec: system}
In this subsection, we give a formal description of the considered model.

\textbf{Influence Model.} The structure of a social network is given by a directed graph $G=(V,E)$ where $V=\{v_1,...,v_n\}$ denotes the set of users and $E$ is edge set. Let $n$ and $m$ be the number of users and edges, receptively. If there is an edge $(u,v)$ in $E$, we say $u$ (resp. $v$) is an in-neighbor (resp. an out-neighbor) of $v$ (resp. $u$). Let $N_v^+$ and $N_v^-$ be the set of the out-neighbors and in-neighbors of $v$, respectively. From the view of influence diffusion, we say a user is active if they successfully receive the target information from their neighbors. Depending on the ways of influence diffusion, many influence models has been developed. In this paper, we consider the general triggering model of which the famous IC and LT models are special cases. 

\begin{definition}[Triggering Model \cite{kempe2003maximizing}]
\label{def:triggering}
An influence model is a triggering model if its diffusion process can be equivalently described as follows. Each node $v$ independently chooses a random ``triggering set" $T_v$ according to some distribution over subsets of its in-neighbors. To start the process, we target a set of users for initial activation. After the initial iteration, an inactive node $v$ become active at time step $t$ if it has an in-neighbor in its chosen triggering set $T_v$ that is active at time step $t-1$.
\end{definition}

By Def. \ref{def:triggering}, a triggering model is given by a pair $(G,\mathcal{D})$ where $G=(V=\{v_1,...,v_n\},E)$ is the network structure and $\mathcal{D}=\{\mathcal{D}_{v_1},...,\mathcal{D}_{v_n}\}$ is a set of distributions. In particular, $\mathcal{D}_{v_i}$ specifies the distribution over subsets of the in-neighbors of user $v_i$. As two special triggering models, IC model and LT model are defined as follows.

\begin{definition}[IC model.] 
The IC model assumes that an active user has one chance to activate each of their neighbors with a certain probability and the activations of different pairs of users are independent. Under this model, each edge $(u,v)$ is assigned a real number $p_{(u,v)} \in (0,1]$ which is the propagation probability from $u$ to $v$. Under the IC model, the distribution $\mathcal{D}_{v_i}$ is given by the probability of the edges. For a node $v$, an in-neighbor $u$ of $v$ has the probability $p_{(u,v)}$ to appear in the triggering set $T_v$.
\end{definition}

\begin{definition}[LT model.] 
Under the general threshold model, a user becomes active if they have received sufficient influence from their neighbors. Specifically, each edge $(u,v)$ has a weight $w_{(u,v)}>0$ and each user $v$ holds a threshold $\theta_v >0$. The user $v$ becomes active if $\sum_{u \in A(v)}w_{(u,v)}\geq \theta_v$ where $A(v)$ is the set of active in-neighbors of user $v$. For the linear threshold model, we further require that, for each user $v$, $\sum_{u \in N_v^-}w_{(u,v)} \leq 1$ and the threshold $\theta_v$ is selected uniformly at random from $[0, 1]$. As shown in \cite{kempe2003maximizing}, taking LT model as a triggering model, the triggering set $T_v$ of a node $v$ is decided as follows: $v$ selects \textit{at most one} of its in-neighbors at random where an in-neighbor $u$ has the probability of $w_{(u,v)}$ to be selected and with probability $1-\sum_{u \in N_v^-}w_{(u,v)}$ that $v$ does not pick any edges.
\end{definition}

Note that the diffusion process under triggering model is stochastic as the triggering sets are generated at random. An important property of triggering model is shown as follows.

\begin{theorem}[\cite{kempe2003maximizing}]
\label{theorem:triggering}
Under a triggering model, the expected number of active nodes is a submodular function with respect to the seed set.
\end{theorem}

For most of the triggering models, the process defined in Def. \ref{def:triggering} is not the real diffusion process. For example, instead of the generating the triggering sets in advance,  the real diffusion process of IC model goes round by round from inactive users to active users until no user can be further activated.

\textbf{Coupon System.} In this paper, we consider the very basic coupon system where a product is associated with two features, $P$ and $C$, which, respectively, denote the price of the product and the value of the coupon. Let $I_v$ be the \textit{intrinsic value}\footnote{The intrinsic value is also referred as costumer evaluation in other literatures. Obtaining the real intrinsic value of a certain customer is out of the scope of this paper and we assume it is given in prior.} held by $v$ for the considered product. After receiving the information of the product, a user becomes an adopter if they believe the product is fairly priced, i.e., $I_v \geq P$. Without coupons, only the users with $I_v \geq P$ can be considered as seed users, while once given a coupon the user $v$ with $I_v \geq P-C$ can also be the seed adopters. In equivalent, one may say the intrinsic value of user $v$ increases to  $I_v+C$ if a coupon is offered. In a social network, there can be some influential users whose intrinsic value is less than the price of the product, but coupon allows the seller to activate such powerful users and consequently raises the total profit. We denote by $r=\frac{P-C}{P}$ the ratio of the price $P$ to $P-C$. One can take $r$ as a normalized discount ratio of the product.  

\textbf{T-C Model.} Combining the influence model with coupon system, we have the triggering-coupon (T-C) model of which the marketing process unfolds as follows.
\begin{itemize}
\item Initially all the nodes are inactive and a seed set $S \subseteq V$ is decided.
\item At each time step $t$, an inactive node $v$ become active if (a) $v$ is activated by its neighbors and (b)  $I_v \geq P$.
\item The process terminates when there is no user can be further activated.
\end{itemize}

Under the T-C model, the active users correspond to the adopters of the product. To distinguish it from the pure influence model, we will use \textit{adopter} instead of active user in the rest of this paper. With the coupon system, the diffusion process is slightly different from the traditional spread of influence. However, the T-C model itself forms a new triggering model.
\begin{lemma}
\label{lemma:triggering}
A triggering model combined with the coupon system yields a new triggering model.
\end{lemma}
\begin{proof}
Let $\mathcal{D}=\{\mathcal{D}_{v_1},...,\mathcal{D}_{v_n}\}$ be the set of the distributions given by the triggering model. With the coupon system, the T-C model can be taken as another triggering model where the new triggering distribution $\mathcal{D}_{v}^{'}$ of each node $v$ is defined as 
\begin{equation}
\label{eq:new_distribution}
 \mathcal{D}_{v}^{'} =
  \begin{cases}
    \mathcal{D}_{v}  &  \hspace{0mm} \hspace{-0.5mm} \text{if $I_v \geq P$}  \\
    \text{empty distribution}  & \hspace{0mm} \hspace{-0.5mm} \text{otherwise.} 
    \end{cases}
\end{equation}

where the empty distribution always returns the empty set. That is, for a node $v$, if $I_v \geq P$ then the distribution remains unchanged and $T_v$ is decided according to $\mathcal{D}_v$, otherwise, $T_v$ is always the empty set as $v$ can never be an adopter. 
\end{proof}

An instance of the T-C model is called a \textit{T-C network}. A T-C network $\Upsilon$ consists of a graph topology,  the price of the product, the value of the coupon, and a distribution of triggering set for each node defined in Eq. (\ref{eq:new_distribution}). In this paper, we use the general T-C model for analysis, and adopt the IC-coupon (IC-C) and LT-coupon (LT-C) models for experiments.

\subsection{Problem Definition}
Given a set $S$ of seed users and a T-C network $\Upsilon$, let $\pi(S)$ be the expected number of adopters under $S$. For a user $v$, if $P > I_v+C$, then $v$ cannot be an adopter even if offered with a coupon. Without loss of generality, we can remove such nodes from the network in advance and therefore assume that $P \leq I_u+C$ holds for every user $u \in V$. For a seed set $S$, the earned profit is $P\cdot \pi(S)$ minus the cost of the offered coupons. Let 
\begin{equation}
\label{eq: objective}
f(S)=P \cdot \pi(S)-C\cdot |S|
\end{equation}
be the profit under $S$. In this paper, we aim to find a seed set $S$ such that the profit can be maximized, which is termed as the Profit Maximization (PM) problem. 
\begin{problem}[Profit Maximization]
Given a T-C network $\Upsilon$, find a set $S$ of seed adopters such that $f(S)$ can be maximized. 
\end{problem}

Let $V_{opt} \subseteq V$ be the optimal solution. Note that $f(V_{opt}) \geq f(V)$ and therefore
\begin{equation}
\label{eq:lower_bound}
f(V)=(P-C)\cdot n
\end{equation}
is a lower bound of the optimal profit. The following lemma gives the key property of our objective function.

\begin{lemma}
\label{lamme:f_sub}
$f(S)$ is submodular but not always monotone increasing. 
\end{lemma}
\begin{proof}
Since $|S|$ is linear, it suffices to show $\pi()$ is submodular. By Lemma \ref{lemma:triggering}, the T-C model is a triggering model and therefore $\pi()$ is submodular due to Theorem \ref{theorem:triggering}.

Consider a simple instance of IC-C model where $V=\{v_1,v_2\}$, $E=\{(v_1,v_2)\}$ and $p_{(v_1,v_2)}=1$. Suppose that $P \geq I_{v_2}$. In this example, $f(\{v_1\})=2P-C$ while $f(\{v_1,v_2\})=2P-2C$. Therefore, $f()$ may not be monotone increasing, which inevitably incurs more difficulty in designing good algorithms. 
\end{proof}

\subsection{Unconstrained Submodular Maximization}
With the recent works, several approximation algorithms are now available for non-monotone submodular maximization. In particular, Feige \textit{et al.} \cite{feige2011maximizing} propose a deterministic local-search $\frac{1}{3}$-approximation and a randomized $\frac{2}{5}$-approximation algorithm for maximizing nonnegative submodular functions. The state-of-the-art algorithm \cite{buchbinder2015tight} achieves a $\frac{1}{2}$-approximation proposed by Buchbinder \textit{et al.}, as shown in Alg. \ref{alg:1/2}.  

\begin{algorithm}[t]
\caption{Buchbinder's Algorithm}\label{alg:1/2}
\begin{algorithmic}[1]
\State \textbf{Input:} $h(), \mathcal{N}=\{v_1,...,v_n\}$; 
\State $X_0 \leftarrow \emptyset$, $Y_0 \leftarrow \mathcal{N}$;
\For {$i= 1:n$}
\State $a_i \leftarrow h(X_{i-1}\cup \{v_i\})-h(X_{i-1})$;
\State $b_i \leftarrow h(Y_{i-1}\setminus \{v_i\})-h(Y_{i-1})$;
\State $a_i^{'} \leftarrow \max\{a_i,0\}$; $b_i^{'} \leftarrow \max\{b_i,0\}$;
\State $rand \leftarrow$ a random number from 0 to 1 generated in uniform; 
\If {$rand  \leq a_i^{'}/(a_i^{'}+b_i^{'})$ or $a_i^{'}+b_i^{'}=0$}
	\State $X_i \leftarrow X_{i-1}\cup \{v_i\}$, $Y_i \leftarrow Y_{i-1}$.
	\Else
	\State $X_i \leftarrow X_{i-1}$, $Y_i \leftarrow Y_{i-1}\setminus \{v_i\}$;
	\EndIf
\EndFor  
\State Return $X_n$;
\end{algorithmic}
\end{algorithm}

\begin{lemma}[\cite{buchbinder2015tight}]
\label{lemma:1/2}
Given an objective function $h()$ over a ground set $\mathcal{N}$, the set produced by Alg. \ref{alg:1/2} is a $\frac{1}{2}$-approximation\footnote{Strictly speaking, Alg. \ref{alg:1/2} provides an expected $\frac{1}{2}$-approximation as it is a randomized algorithm. For simplicity, we will take such a solution as an exact $\frac{1}{2}$-approximation. This does not affect the fundamental issues of our analysis. } provided that $h()$ is submodular and $h(\emptyset)+h(\mathcal{N})\geq 0$.\footnote{In \cite{buchbinder2015tight}, it is assumed that $h()$ is non-negative while in fact we only need that $h(\emptyset)+h(\mathcal{N})\geq 0$. The proof is similar to that of Theorem \ref{theorem:accuracy} shown later. A detailed discussion is provided by J. Tang \textit{el al.} \cite{tang2016profit}.}
\end{lemma}

The existing algorithms assume the availability of an oracle of the objective function. However, it turns out that computing the exact value of $\pi(S)$ is \#P-hard \cite{chen2010scalable}. As seen in the recent works on influence related problems, the main difficulty in solving such problems has shifted from approximation algorithm design to the estimating of the expected influence $\pi(S)$. Since the value of $f()$ is hard to compute, Alg. \ref{alg:1/2} cannot be directly applied. Given that the diffusion process is stochastic, for a certain seed set $S$, $f(S)$ can be naturally estimated by sampling. In this paper, we discuss several sampling methods and show how they can be integrated into the existing approximation algorithms. In particular, we use the Buchbinder's algorithm (Alg. \ref{alg:1/2}) for illustration. For a small value $\epsilon\in (0, \frac{1}{2})$ and  a large value $N >0$, an algorithm is a $(\frac{1}{2}-\epsilon,1-\frac{1}{N})$-approximation if it is able to produce a solution with an approximation ratio of $\frac{1}{2}-\epsilon$ with least $1-1/N$ probability \footnote{$\frac{1}{2}$ is the best possible ratio for general unconstrained submodular maximization \cite{feige2011maximizing}.}. We aim to design the algorithms with such a performance guarantee and assume that $\epsilon$ and $N$ are fixed in the rest of this paper.

Sampling methods can be generally classified into two categories, forward sampling and reverse sampling. As the name suggests, the forward sampling methods obtain an estimate of $\pi(S)$ by directly running the diffusion process while the reverse sampling methods simulate the diffusion process in the reverse direction. The algorithms designed based on such sampling techniques are shown in next two sections.

\section{Forward Sampling}
\label{sec:front}
Before discussing the PM problem, let us first consider the general submodular maximization problem in the case that the exact function value is hard to compute. For Alg. \ref{alg:1/2}, suppose that an estimate $\tilde{h}()$ of $h()$ is employed whenever $h()$ is called. One can expect that the approximation ratio can be arbitrarily close to $\frac{1}{2}$ as long as the estimates are sufficiently accurate. The modified algorithm is shown in Alg. \ref{alg:1/2_estimate}, which slightly differs from Alg. \ref{alg:1/2} and takes a lower bound $L^*$ of the optimal value as an extra input. The following key theory indicates the relationship between the precision of the estimates and the approximation ratio.

\begin{algorithm}[t]
\caption{Forward Sampling Framework}\label{alg:1/2_estimate}
\begin{algorithmic}[1]
\State \textbf{Input:} $(\tilde{h}(), \mathcal{N}=\{v_1,...,v_n\},\epsilon,L^{*})$; 
\State $X_0 \leftarrow \emptyset$, $Y_0 \leftarrow \mathcal{N}$;
\For {$i= 1:n$}
\State $\tilde{a}_i \leftarrow \tilde{h}(X_{i-1}\cup \{v_i\})-\tilde{h}(X_{i-1})+\frac{2\epsilon }{n}\cdot L^{*}$;
\State $\tilde{b}_i \leftarrow \tilde{h}(Y_{i-1}\setminus \{v_i\})-\tilde{h}(Y_{i-1})+\frac{2\epsilon }{n}\cdot L^{*}$;
\State $\tilde{a}_i^{'} \leftarrow \max\{\tilde{a}_i,0\}$; $\tilde{b}_i^{'} \leftarrow \max\{b_i,0\}$;
\State $rand \leftarrow$ a random number from 0 to 1 generated in uniform; 
\If {$rand  \leq \tilde{a}_i^{'}/(\tilde{a}_i^{'}+\tilde{b}_i^{'})$ or $\tilde{a}_i^{'}+\tilde{b}_i^{'}=0$}
	\State $X_i \leftarrow X_{i-1}\cup \{v_i\}$, $Y_i \leftarrow Y_{i-1}$.
	\Else
	\State $X_i \leftarrow X_{i-1}$, $Y_i \leftarrow Y_{i-1}\setminus \{v_i\}$;
	\EndIf
\EndFor  
\State Return $X_n$;
\end{algorithmic}
\end{algorithm}

\begin{theorem}
\label{theorem:accuracy}
For any instance $(h(),\mathcal{N}=\{v_1,...,v_n\})$ and any $\epsilon \in (0,\frac{1}{2})$, Alg. \ref{alg:1/2_estimate} is a $(\frac{1}{2}-\epsilon)$-approximation, if 
\begin{equation}
\label{eq:accuracy}
|\tilde{h}(S)-h(S)| \leq \frac{\epsilon}{n}\cdot L^{*}
\end{equation}
holds for each $S$ inspected by Alg. \ref{alg:1/2_estimate} where $\tilde{h}(S)$ is the estimate of $h(S)$ and $L^{*}$ is a lower bound of the optimal value of $h()$.
\end{theorem}
\begin{proof}
See Appendix \ref{sebsec:proof_theorem:accuracy}.
\end{proof}

Alg. \ref{alg:1/2_estimate} provides a framework for maximizing a submodular function $h()$ when the exact value of $h()$ is hard to compute and an accurate estimator $\tilde{h}()$ is obtainable. In the rest of this section, we introduce two estimators of $\pi()$ and apply the above framework to solve the PM problem. 

\subsection{SPM Algorithm}
\label{subsec:simu}
A straightforward way to estimate $\pi(S)$ is to use Monte Carlo simulation. That is, for a given seed set $S$, we simulate the diffusion process under the considered T-C model to obtain the samples of $\pi(S)$ and take the sample mean as the estimator. For a seed set $S \subseteq V$, let $\widetilde{\pi}_l(S)$ be the sample mean obtained by running $l$ times of simulations and define that $\widetilde{f}_l(S)=P\cdot \widetilde{\pi}_l(S)-C\cdot|S|$.  By taking $f(V)$ as the lower bound of $f(V_{opt})$, the Alg. \ref{alg:1/2_estimate} with input $(\tilde{f}_l(),V,\epsilon,f(V))$ is able to produce a good approximation to the PM problem provided that Eq. (\ref{eq:accuracy}) is satisfied for each estimating. We denote by this approach as the simulation-based profit maximization (SPM) algorithm. 

For each estimating, the difference between $\widetilde{\pi}_l(S)$ and $\pi(S)$ can be bounded by Chernoff bound as the simulations are executed independently. A useful form of Chernoff bound is provided in Appendix \ref{appendix:1} and we will use Eqs. (\ref{eq:chernoff_1}) and (\ref{eq:chernoff_2}) throughout this paper. The following lemma shows the number of simulations needed to meet the accuracy required by Eq. (\ref{eq:accuracy}). 
\begin{lemma}
\label{lemma:1_1}
For a certain set $S \subseteq V$ and $\frac{1}{2}>\epsilon >0$, we have
\begin{equation}
\Pr[|\widetilde{f}_l(S)-f(S)| > \frac{\epsilon}{n} f(V)] \leq \frac{1}{4nN},
\end{equation}
if $l \geq \delta_0$ where
\begin{equation}
\label{eq:delta_0}
\delta_0=\frac{(\ln 8+\ln n+ \ln N)(2n^2+\epsilon r n)}{\epsilon^2 r^2}.
\end{equation}
\end{lemma}
\begin{proof} It can be proved by directly applying the Chernoff bound. By rearrangement, $\Pr[|\widetilde{f}_l(S)-f(S)| > \frac{\epsilon}{n} \cdot f(V)]$ is equal to $\Pr[|l\cdot \frac{\widetilde{\pi}_l(S)}{n}-l\cdot \frac{\pi(S)}{n}| > l \cdot  \frac{\pi(S)}{n}\cdot  \frac{\epsilon  (P-C)}{P\cdot \pi(S)}]$. Applying Chernoff bound, this probability is no larger than $$2\exp(-\frac{l \cdot \frac{\pi(S)}{n} \cdot (\frac{\epsilon(P-C)}{P\cdot \pi(S)})^2}{2+(\frac{\epsilon(P-C)}{P \cdot \pi(S)})})$$. By Eq. (\ref{eq:delta_0}) and $\pi(S)\leq n$, this probability is at most $\frac{1}{4nN} $.

\end{proof}

\begin{theorem}
\label{theorem:1}
For any $\epsilon \in (0,\frac{1}{2})$ and $N >0$, to achieve a $(\dfrac{1}{2}-\epsilon)$-approximation with probability at least $1-1/N$, the SPM algorithm demands the running time of $O(mn^3\ln n)$.
\end{theorem} 

\begin{proof}
Since there are $4n$ sets inspected by Alg. \ref{alg:1/2_estimate}, by Lemma \ref{lemma:1_1} and the union bound, $|\widetilde{f}_l(S)-f(S)| \leq \frac{\epsilon}{n} f(V)$ holds for every inspected set with probability at least $1-1/N$ provided that $l \geq \delta_0$. Therefore, by Theorem \ref{theorem:accuracy}, the approximation ratio is $\dfrac{1}{2}-\epsilon$ if $l = \delta_0$.  Because $\delta_0$ simulations are used for each estimating and there are totally $4n$ estimated function values, the number of executed simulations is $4\delta_0n$. The total running time is therefore $O(4\delta_0 n m)$ as one simulation costs $O(m)$ in the worst case. 
\end{proof}

Note that executing one simulation for $f(S)$ costs $O(m)$ in the worst case but in fact it is proportional to $f(S)$ which is usually much smaller than $m$. 

\subsection{RPM Algorithm}
\label{subsec:realization}
Now let us consider another estimator of $f()$. As mentioned in Sec. \ref{sec: pre}, the real diffusion process of a triggering model, namely IC and LT model, is equivalent to the one described in Def. \ref{def:triggering}. Therefore, instead of simulating the real diffusion process, one can obtain an estimate of $\pi(S)$ by sampling the triggering set of each node. When the triggering set of each node is determined, we obtain a \textit{realization} of the given network. 
\begin{definition}[Realization]
\label{def:realization}
For a T-C network $\Upsilon$, a realization $g=\{T_1,...,T_n\}$ is a collections of triggering sets of the nodes where $T_v$ is sampled according to distribution given by $\Upsilon$. Equivalently, a realization can be taken as a deterministic T-C network. 
\end{definition}


\begin{algorithm}[t]
\caption{Realization-based Profit Maximization}\label{alg:simu}
\begin{algorithmic}[1]
\State \textbf{Input:} $\Upsilon,\epsilon$ and $N$;
\State $l \leftarrow \frac{(\ln 8+\ln n+ \ln N)(2n^2+\epsilon r n)}{\epsilon^2 r^2}$;
\State Generate $l$ realizations $\mathcal{G}_l=\{g_1,...,g_l\}$ by sampling triggering sets for each node;
\State $V^{*} \leftarrow$ Forward Sampling Framework$(\widehat{f}_{\mathcal{G}_l}(),V,\epsilon,f(V))$;
\State Return $V^{*}$;
\end{algorithmic}
\end{algorithm}

The following lemma directly follows from the definition of triggering model.

\begin{lemma}
\label{lemma:equivalent} 
Given a seed set $S$, the following two processes are equivalent to each other with respect to $\pi(S)$.
\begin{itemize}
\item Process a. Execute the real diffusion process on $G$.
\item Process b. Randomly generate a realization $g$ and execute the diffusion process on $g$ according to Def. \ref{def:triggering}.
\end{itemize}
\end{lemma}
 
For a seed set $S$, let $\widehat{\pi}_g(S)$ be the number of adopters under $S$ on realization $g$ and correspondingly we define that $\widehat{f}(g,S)=P\cdot \widehat{\pi}_g(S)-C\cdot |S|$. According Lemma \ref{lemma:equivalent}, one can estimate $\pi(S)$ by repeatedly executing the Process b in Lemma \ref{lemma:equivalent} and calculating the mean of $\widehat{\pi}_g(S)$. Given a set of $l$ realizations $\mathcal{G}_l=\{g_1,...,g_l\}$ generated independently, for a seed set $S$, let
\begin{equation}
\widehat{\pi}_{\mathcal{G}_l}(S)=\frac{\sum_{i=1}^{l}\widehat{\pi}_{g_i}(S)}{l}
\end{equation}
and
\begin{equation}
\widehat{f}_{\mathcal{G}_l}(S)=\frac{\sum_{i=1}^{l}\widehat{f}(g_i,S)}{l}=P \cdot \widehat{\pi}_{\mathcal{G}_l}(S)-C\cdot |S|.
\end{equation}
Now we have another estimator $\widehat{f}_{\mathcal{G}_l}()$ that can be utilized in the framework of Alg. \ref{alg:1/2_estimate}. The obtained algorithm is shown in Alg. \ref{alg:simu} which is named as \textit{realization-based profit maximization} (RPM) algorithm. According to Lemma \ref{theorem:accuracy}, in order to obtain a good approximation one should generate sufficient number of realizations such that Eq. (\ref{eq:accuracy}) can be satisfied. In fact, the number of total realizations required herein is the same as that of the simulations required in the SPM algorithm, shown as follows.  

\begin{lemma}
Suppose $l$ realizations are generated for estimating. For any $\epsilon \in (0,\frac{1}{2})$ and $N >0$, with probability at least $1-\frac{1}{N}$, $|\widehat{f}_{\mathcal{G}_l}(S)-f(S)|\leq \frac{\epsilon}{n} f(V)$ holds for
every $S$ inspected by Alg .\ref{alg:1/2_estimate}, if $l$ is not less than $\delta_0$ defined in Eq. (\ref{eq:delta_0}).
\end{lemma}
\begin{proof}
The proof is same as that of Lemma \ref{lemma:1_1}.
\end{proof}

According to Def. \ref{def:triggering}, generating one realization costs $O(m)$ time. For a certain realization $g$ and a seed set $S$, computing $\widehat{f}_(g,S)$ can be done in $O(n+m)$ time by the breadth-first search, and therefore $\widehat{f}_{\mathcal{G}_l}(S)$ can be computed in $O(l(m+n))$. Hence, Alg. \ref{alg:1/2} totally costs $O(lm+ln(m+n))$. Setting $l$ to be $\delta_0$ gives the running time of $O(mn^3\ln n)$.
\begin{theorem}
\label{theorem:1/2}
For any $\epsilon \in (0,\frac{1}{2})$ and $N >0$, the RPM algorithm provides a $(\dfrac{1}{2}-\epsilon)$-approximation with a probability at least $1-1/N$ with the running time of $O(mn^3\ln n)$ with respect to $m$ and $n$.
\end{theorem} 

One can see that the running times of SPM algorithm and RPM algorithm are in the same order with respect to $m$ and $n$. However, the actual number of operations executed by SPM algorithm is about $4 l m n$ while the RPM algorithm roughly conducts $l m n$ operations. Therefore, when the graph is large, the SPM algorithm is likely to be more efficient. For practical implementation, because one realization $g$ can be used to estimate $f(S)$ for any seed set $S$, we can generate the realizations in advance and use them later in real-time computation, which is another advantage of the RPM algorithm.

\section{Reverse Sampling}
\label{sec:reverse}
In this section, we present two algorithms which are able to provide the approximation ratio of $\frac{1}{2}-\epsilon$ but require significantly less running time than RPM and SPM do. Such algorithms are designed based on the technique of reverse sampling which is firstly invented by Borgs. \textit{et al.} \cite{borgs2014maximizing} for the IM problem. In the reserve sampling framework, the key object is the \textit{reverse adopted-reachable set}.

\begin{algorithm}[t]
\caption{ RA Set }\label{alg:ra_set}
\begin{algorithmic}[1]
\State \textbf{Input:} $\Upsilon$; 
\State Select a node $v$ from $V$ uniformly in random; 
\State Sample a realization $g$ of $\Upsilon$;
\State Return the nodes that are reverse-adopted-reachable to $v$ in $g$;
\end{algorithmic}
\end{algorithm}

\begin{definition}
For a node $v$ and a realization $g=\{T_1,...,T_n\}$ of a T-C network $\Upsilon$, a node $v^*$ is \textit{reverse-adopted-reachable} to $v$ in $g$ if there is a sequence of nodes $u_1,...,u_k$ such that $u_1=v, u_k=v^*$ and $v_{i+1} \in T_{v_i}$ for $1 \leq i \leq k-1$.
\end{definition}

\begin{definition}{\small(\textbf{Reverse Adopted-reachable (RA) Set}\footnote{The RA set defined in this paper is analogous to the hyperedge in \cite{borgs2014maximizing}, the reverse reachable set in \cite{tang2015influence} and the random R-tuple in \cite{tong2017efficient}.})} 
\label{def:raset}
For a T-C network $\Upsilon$, an RA set is a set of nodes randomly generated by Alg. \ref{alg:ra_set}. As shown therein, we first select a node $v$ in $V$ uniformly in random and then sample a realization $g$ of $\Upsilon$. Finally, it returns the nodes that are reverse-adopted-reachable of $v$ in $g$.
\end{definition}

For a seed set $S \subseteq V$ and a RA set $R$, let $x(R,S)$ be a variable such that
\[
 x(S,R) =
  \begin{cases}
  1  &  \hspace{0mm} \hspace{-0.5mm} \text{if $R \cap S \neq \emptyset$}  \\
  0  & \hspace{0mm} \hspace{-0.5mm} \text{otherwise.} 
  \end{cases}
\] 
Since $R$ is generated randomly, $x(S,R)$ is a random variable. The following lemma shows that $n\cdot x(S,R)$ is a unbiased estimate of $\pi(S)$.
\begin{lemma}
\label{lemma:ra_mean}
For any seed set $S$, $\mathbb{E}[x(S,R)]=\frac{\pi(S)}{n}$.
\end{lemma}
\begin{proof}
See Appendix \ref{sec:proof_lemma:ra_mean}
\end{proof}

Let $\mathcal{R}_l=\{R_1,..,R_l\}$ be a set of $l$ RA sets \textit{independently} generated by Alg. \ref{alg:ra_set} and let 
\begin{equation*}
F(\mathcal{R}_l,S)=P\cdot n\cdot \frac{\sum_{i=1}^{l}x(R_i,S)}{l}-C\cdot|S|.
\end{equation*}

The following problem plays an important role in solving the PM problem.

\begin{problem}[Cost Set Cover Problem]
\label{problem:cost_set_cover}
Given a collection $\mathcal{R}_l$ of RA sets, find a set $S$ such that $F(\mathcal{R}_l,S)$ is maximized.
\end{problem}
\begin{lemma}
Given $\mathcal{R}_l$, $F(\mathcal{R}_l,S)$ is a submodular function with respect to $S$.
\end{lemma}
\begin{proof}
The first part of $F(\mathcal{R}_l,S)$ is exactly the set cover problem and its second part is linear to $|S|$. 
\end{proof}
Since $F(\mathcal{R}_l,S)$ is submodular and $F(\mathcal{R}_l,\emptyset)+F(\mathcal{R}_l,V)=0+(P-C)\cdot n \geq 0$, by Lemma \ref{lemma:1/2}, Alg. \ref{alg:1/2} provides an $\frac{1}{2}$-approximation to Problem \ref{problem:cost_set_cover}.

By Lemma \ref{lemma:ra_mean}, $F(\mathcal{R}_l,S)$ is an unbiased estimate of $f(S)$. Hence, the set $S$ that can maximize $F(\mathcal{R}_l,S)$ can intuitively be a good solution to the PM problem, which is the main idea to solve the PM problem by the reverse sampling.

\subsection{RA-T Algorithm}
\label{subsec:RA-T}
In this section, we present the first algorithm designed with the reverse sampling, as shown in Alg. \ref{alg:reverse_1}. In this algorithm, we first decide a threshold $l$, generate $l$ RA sets, and then run Alg. \ref{alg:1/2} to solve Problem \ref{problem:cost_set_cover}. We term this algorithm as the RA-T algorithm where RA-T stands for reverse adopted-reachable (RA) set with a threshold (T). As shown in the following, for any $\epsilon \in (0,\frac{1}{2})$, the RA-T algorithm produces a $(\frac{1}{2}-\epsilon)$-approximation to the PM problem with a high probability, as long as $l$ is sufficiently large. 

\begin{algorithm}[t]
\caption{RA-T Algorithm}\label{alg:reverse_1}
\begin{algorithmic}[1]
\State \textbf{Input:} $\Upsilon, \epsilon$ and $N$; 
\State $\delta_1 \leftarrow \dfrac{(\ln N+n\ln 2)(2+\epsilon_1r)}{\epsilon_1^2r^2}$; 
\State $\delta_2 \leftarrow \dfrac{2\ln N}{\epsilon_2^2r^2}$;
\State $l=\max(\delta_1,\delta_2)$;  
\State Independently generate a collection $\mathcal{R}_l=\{R_1,...,R_l\}$ of $l$ RR sets by Alg. \ref{alg:ra_set}; 
\State $V^* \leftarrow$ Buchbinder's Algorithm ($F(\mathcal{R}_l,S),V$);
\State Return $V^{*}$;
\end{algorithmic}
\end{algorithm} 

Let $\epsilon_1>0$ and $\epsilon_2>0$ be some adjustable parameters where
\begin{equation}
\label{eq:ep1+ep2}
\epsilon_1+\frac{1}{2}\epsilon_2=\epsilon.
\end{equation}
Set that
\begin{equation}
\label{eq:delta_1}
\delta_1=\dfrac{(\ln N+n\ln 2)(2+\epsilon_1r)}{\epsilon_1^2r^2},
\end{equation}
and
\begin{equation}
\label{eq:delta_2}
\delta_2=\dfrac{2\ln N}{\epsilon_2^2r^2},
\end{equation}

\begin{lemma}
\label{lemma:epsilon_1}
Given a set $\mathcal{R}_l$ of $l$ RA sets generated independently by Alg. \ref{alg:ra_set}, with probability at most $1/N$, there exists some $S \subseteq V$ such that
\begin{equation*}
F(\mathcal{R}_l,S)-f(S) > \epsilon_1 \cdot f(V_{opt}),
\end{equation*}
if $l \geq \delta_1$.
\end{lemma}
\begin{proof}
See Appendix \ref{sec:proof_lemma:epsilon_1}.
\end{proof}

\begin{lemma}
\label{lemma:epsilon_2}
Given a set $\mathcal{R}_l$ of $l$ RA sets generated by Alg. \ref{alg:ra_set}, with probability at most $1/N$
\begin{equation*}
F(\mathcal{R}_l,V_{opt})-f(V_{opt}) < -\epsilon_2 \
\cdot f(V_{opt})
\end{equation*}
if $l \geq \delta_2$.
\end{lemma}
\begin{proof}
See Appendix \ref{sec:proof_lemma:epsilon_2}.
\end{proof}

For the RA-T algorithm, intuitively we should assure that $F(\mathcal{R}_l,V_{opt})$ and $F(\mathcal{R}_l,V^{*})$ are close to $f(V_{opt})$ and $f(V^{*})$, respectively, where $V^{*}$ is obtained by solving Problem \ref{problem:cost_set_cover}. As shown in Lemma \ref{lemma:epsilon_2}, $F(\mathcal{R}_l,V_{opt})$ is sufficient accurate if $l \geq \delta_2$. Because $V^{*}$ is unknown in advance, $l$ has to be large enough such that the difference between $F(\mathcal{R}_l,S)$ and $f(S)$ can be bounded for every $S$, and therefore $\delta_1$ is larger than $\delta_2$ by a factor of $n$. When $F(\mathcal{R}_l,V_{opt})$ and $F(\mathcal{R}_l,V^{*})$ are both accurate, the RA-T algorithm is guaranteed to be effective, as shown in the next lemma.

\begin{lemma}
\label{lemma:ratio}
Let $S^{*}$ be the set returned by Alg. \ref{alg:1/2} with input ($F(\mathcal{R}_l,S),V$). For any $\epsilon \in (0,\frac{1}{2})$ and $N \geq 0$,
\begin{equation*}
f(V^{*}) \geq (\frac{1}{2}-\epsilon)f(V_{opt})
\end{equation*}
holds with probability at least $1-\frac{2}{N}$ provided that
\begin{equation*}
l = \max(\delta_1,\delta_2) 
\end{equation*}
\end{lemma}
\begin{proof}
By Lemma \ref{lemma:epsilon_1}, $f(V^{*}) \geq F(\mathcal{R}_l,V^{*})-\epsilon_1f(V_{opt})$. By Lemma \ref{lemma:1/2}, $F(\mathcal{R}_l,V^{*}) \geq \frac{1}{2}F(\mathcal{R}_l,V_{opt})$ and therefore, $f(V^{*}) \geq \frac{1}{2}F(\mathcal{R}_l,V_{opt})-\epsilon_1f(V_{opt})$. Finally, combining Lemma \ref{lemma:epsilon_2} and Eq. (\ref{eq:ep1+ep2}), $f(V^{*}) \geq (\frac{1}{2}-\epsilon)f(V_{opt})$. According to the union bound, the above inequality holds with probability at least $1-2/N$.
\end{proof}

\textbf{Running time.} Now let us consider the running time of Alg. \ref{alg:reverse_1}. Given $\mathcal{R}_l$, calculating $F(\mathcal{R}_l,S)$ requires $O(ln)$ time. Therefore, line 6 takes $O(ln^2)$ time. Because generating one RA set takes $O(m)$, the running time of Alg. \ref{alg:reverse_1} is $O(lm+ l n^2)$.
Since $l=\max(\delta_1,\delta_2)$, $l=O(n)$. The analysis in this section is summarized as follows.
\begin{theorem}
\label{theorem:alg_2}
Given the adjustable parameters $\epsilon$ and $N$, with probability at least $1-2/N$, Alg. \ref{alg:reverse_1} returns a ($\frac{1}{2}-\epsilon$)-approximation running in $O(n^3)$ with respect to $n$.
\end{theorem}

The success probability can be increased to $1-1/N$ by scaling $N$, which yields a $(\frac{1}{2}-\epsilon, 1-\frac{1}{N})$-approximation to the PM problem. Note that the running time of Alg. \ref{alg:reverse_1} can be further reduced. Given $n, r, N$ and $\epsilon$, the best choice of $\epsilon_1$ and $\epsilon_2$ is the one that can minimize $\max(\delta_1, \delta_2)$. That is,
\begin{equation}
\label{eq:optimal}
(\epsilon_1,\epsilon_2)=\argmin_{\epsilon_1+\frac{1}{2}\epsilon_2=\epsilon}\max(\delta_1, \delta_2).
\end{equation}
In practice, there is no need to solve Eq. (\ref{eq:optimal}) directly, and, instead one can enumerate $\epsilon_1$ from $0$ to $1/2$ with a small search step, namely $0.01$. We will adopt such a setting in experiments. 

\textbf{Generating RA set.} It is worthy to note that, according to Def. \ref{def:raset}, there is no need to sample the whole realization to obtain an RA set. Instead, one can obtain the adopted-reachable nodes along with the sampling of realization from the node selected in line 2 of Alg. \ref{alg:ra_set} until no nodes can be further reached. Generating a RA set in this way can reduce the running time in practice but does not affect the worst-case running time of RA-T.

\subsection{RA-S Algorithm}
\label{subsec:RA-S}

In this section, we present another algorithm. Recall that in the RA-T algorithm, the number of generated RA sets should be large enough such that $|F(\mathcal{R}_l,V^{*})-f(V^{*})|$ and $|F(\mathcal{R}_l,V_{opt})-f(V_{opt})|$ are sufficiently small where $V^{*}$ is obtained by solving Problem \ref{problem:cost_set_cover}. In this section, we give another algorithm which generates enough RA sets such that $|F(\mathcal{R}_l,V_{opt})-f(V_{opt})|$ is small but controls the error $|F(\mathcal{R}_l,V^{*})-f(V^{*})|$ by Monte Carlo simulations. We call this algorithm as the RA-S algorithm where RA and S stand for \textit{RA set} and \textit{simulation}, respectively.  

\begin{algorithm}[t]
\caption{RA-S Algorithm}\label{alg:reverse_2}
\begin{algorithmic}[1]
\State \textbf{Input:} $\Upsilon, N, \epsilon, k$ and $\epsilon_3$; 
\State Obtain $(\delta_1^*,\delta_2^*,\epsilon_1,\epsilon_2)$ by solving Eq System \ref{eqs:1}.
\State $l=\delta_2^*$, $l^*=\delta_3$ and $\mathcal{R}_l^*=\emptyset$; 
\While{$l \leq 2\cdot \delta_1^*$}
\State Generate $l$ RA sets by Alg. \ref{alg:ra_set} and add them into $\mathcal{R}_l^*$;
\If {$l \geq \delta_1^*$}
\State $V^* \leftarrow$ Buchbinder's Algorithm ($F(\mathcal{R}_l^*,S),V$);
\State Return $V^*$;
\Else
\State $V^* \leftarrow$ Buchbinder's Algorithm ($F(\mathcal{R}_l^*,S),V$);
\State Obtain $\widetilde{f}_{l^*}(V^*)$ by $l^*$ simulations;
\If {$F(\mathcal{R}_l^*,V^*) \leq (1+\epsilon_3)\cdot \widetilde{f}_{l^*}(V^*)$}
\State Return $V^*$;
\EndIf
\State $l=2*l$;
\EndIf
\EndWhile
\end{algorithmic}
\end{algorithm}

The RA-S algorithm is formally shown in Alg. \ref{alg:reverse_2}. In this algorithm, the parameters are $\epsilon_1, \epsilon_2, \epsilon_3, \delta_1^*, \delta_2^*$ and $\delta_3$, where $\epsilon_3$ and $k$ are adjustable and other parameters are determined accordingly. In particular, $\epsilon_1, \epsilon_2, \delta_1^*$ and $\delta_2^*$ are obtained by solving the following equations,
\begin{eqs}
\label{eqs:1}
\begin{align}
&\frac{1-\epsilon_2}{2(1+\epsilon_3)}-\epsilon_1=\frac{1}{2}-\epsilon \nonumber\\
&\delta_1^*=\dfrac{(\ln N+n\ln 2)(6+2\epsilon_1r)}{3\epsilon_1^2r^2} \nonumber\\
&\delta_2^*=\dfrac{2\ln N}{\epsilon_2^2r^2} \nonumber\\
&\delta_1^* = 2^k\cdot \delta_2^* \nonumber\\
&\epsilon_1, \epsilon_2 >0 \nonumber
\end{align}
\end{eqs}

and 
\begin{equation}
\label{eq:delta_3}
\delta_3=\dfrac{(2+\epsilon_1r)\ln N}{\epsilon_1^2r^2},
\end{equation}
In the above equation system, $\epsilon$, $N$, $n$, $r$, $k$ and $\epsilon_3$ are given in advance. Because $\delta_1^*$ (resp. $\delta_2^*$) approaches infinity when $\epsilon_1$ (resp. $\epsilon_2$) approaches zero, there must be at least one set of variables that solves the equations. As shown in Alg. \ref{alg:reverse_2}, with such parameters, we generate a family of collections of RA sets with sizes $\delta_2^*$, $2\cdot \delta_2^*$, ..., $\delta_1^* = 2^k\cdot \delta_2^*$ by a while loop from line 4 to 14. Within each iteration, we first generate a bunch of RA sets and then get a set $V^{*}$ by solving Problem \ref{problem:cost_set_cover}. Next, we obtain an estimate $\widetilde{f}_{l^*}(V^*)$ of $f(V^*)$ by $l^*$ simulations, and check if $F(\mathcal{R}_l,V^{*})$ is smaller than $\widetilde{f}_{l^*}(V^*)=\delta_3$ by a factor of $(1+\epsilon_3)$.

\subsubsection{Performance Guarantee.}As shown in Alg. \ref{alg:reverse_2}, the RA-S algorithm terminates either at line 8 or 13. 

\textbf{Case a.} Suppose that it terminates at line 13 and let $\mathcal{R}_l^*=\{R_1,...,R_l\}$ be the collection of RA sets used in the last iteration. According to line 5, a half of the RA sets in $\mathcal{R}_l^*$ are freshly generated in the last iteration and whether they can be generated depends on the other half of the RA sets according to the condition in line 12, which means the random variables $x(R,V_{opt})$ for $R \in \mathcal{R}_l^*$ are not independent. Therefore, the Chernoff bounds are not applicable. However, such dependency between the RA sets does not severely hurt the convergence of the sample mean and it turns out that the random variables $x(R_1,V_{opt}),...,x(R_l,V_{opt})$ in fact form a martingale \cite{tang2015influence}. Therefore, we can alternatively use concentration inequalities for martingales as shown in Appendix \ref{appendix:1}. By using martingale, the accuracy of  $F(\mathcal{R}_l^*,V_{opt})$ can be ensured as shown in the following lemma. 
\begin{lemma}
\label{lemma:ra_s_epsilon_2}
If $l \geq \delta_2^*$
\begin{equation}
\label{eq:new_opt}
F(\mathcal{R}_l^*,V_{opt})-f(V_{opt}) \geq -\epsilon_2 \cdot f(V_{opt})
\end{equation}
holds with probability at least $1-1/N$. 
\end{lemma}
\begin{proof}
See Appendix \ref{sec:proof_lemma:ra_s_epsilon_2}.
\end{proof}
Furthermore, by setting $l^*$ as $\delta_3$, $\widetilde{f}_{l^*}(V^*)$ is sufficiently accurate.

\begin{lemma}
\label{lemma:ra_s_epsilon_1}
Let $V^{*}$ be the set produced in line 11 and $\widetilde{f}_{l^*}(V^*)$ be the estimate of $f(V^*)$ obtained by $l^*$ simulations where $l^*=\delta_3$. With probability at least $1-1/N$,
\begin{equation}
\label{eq:f_new}
\widetilde{f}_{l^*}(V^*)-f(V^*) \leq \epsilon_1 \cdot f(V_{opt})
\end{equation}
\end{lemma}
\begin{proof}
See Appendix \ref{sec:proof_lemma:ra_s_epsilon_1}.
\end{proof}

Now we are ready to show the approximation ratio. Suppose that Eqs. (\ref{eq:new_opt}) and (\ref{eq:f_new}) hold simultaneously. By Eq. (\ref{eq:f_new}), $(1+\epsilon_3)\big(f(V^*)+\epsilon_1 \cdot f(V_{opt})\big) 
\geq (1+\epsilon_3)\widetilde{f}_{l^*}(V^*)$. According to line 12 in Alg. \ref{alg:reverse_2}, $(1+\epsilon_3)\widetilde{f}_{l^*}(V^*) \geq F(\mathcal{R}_l^*,V^{*})$ and, by Lemma \ref{lemma:1/2}, $F(\mathcal{R}_l^*,V^{*}) \geq \frac{1}{2}F(\mathcal{R}_l^*,V_{opt})$. Therefore, $(1+\epsilon_3)\big(f(V^*)+\epsilon_1 \cdot f(V_{opt})\big) \geq \frac{1}{2}F(\mathcal{R}_l^*,V_{opt}) \geq \frac{1}{2}(1-\epsilon_2)f(V_{opt})$, where the section inequality follows from Eq. (\ref{eq:new_opt}). 
Rearranging the above inequality yields
\begin{equation*}
f(V^*) \geq (\frac{1-\epsilon_2}{2(1+\epsilon_3)}-\epsilon_1)f(V_{opt}).
\end{equation*}
Since, the parameters satisfy Eq System \ref{eqs:1}, we have
\begin{equation*}
f(V^*) \geq (\frac{1}{2}-\epsilon)f(V_{opt}).
\end{equation*}

\textbf{Case b.} If Alg. \ref{alg:reverse_2} terminates at line 8, then $l \geq \delta_1^{*}$ and it is sufficiently large to guarantee a $(\frac{1}{2}-\epsilon,1-\frac{2}{N})$-approximation, as shown in the following lemma.
\begin{lemma}
\label{lemma:ra_s_case_b}
If Alg. \ref{alg:reverse_2} terminates at line 6 returning a seed set $V^{*}$, then 
\begin{equation*}
f(V^{*}) \geq (\frac{1}{2}-\epsilon)f(V_{opt})
\end{equation*}
holds with probability at least $1-2/N$.
\end{lemma}
\begin{proof}
The proof of this lemma is similar to that of the analysis of the RA-T algorithm, except that we use martingale inequalities instead of Chernoff bound. Intuitively, making $l\geq \delta_1^*$ ensures that Alg. \ref{alg:reverse_2} returns a good approximation in case that the stopping criterion in line 12 is never met. See. Appendix \ref{sec:proof_lemma:ra_s_case_b} for the complete proof.
\end{proof}
By the above analysis, the RA-S algorithm has the same approximation ratio as that of the algorithms proposed in the previous sections. Because the number of RA sets is doubled in each iteration until it exceeds $\delta_1^*$, the running of RA-S algorithm is at most larger than that of the RA-T algorithm by a constant factor. Therefore, it is still $O(n^3)$ with respect to $n$ and $m$.

\textbf{Parameter setting.} As mentioned earlier, there are two free parameters, $k$ and $\epsilon_3$. Because $l$ ranges from $\delta_2^*$ to $\delta_1^*$ and $\delta_1^* = 2^k\cdot \delta_2^*$, $k$ is number of iterations executed in Alg. \ref{alg:reverse_2} in the worst case. When $k$ is large, the RA-S algorithm has to execute many iterations in the worst case but it also makes it possible that the RA-S algorithm can terminate at line 13 when $l$ is much smaller than $\delta_1^*$. On the other hand, when $k$ is small, the number of RA sets generated in each iteration becomes larger, but the worst case requires fewer iterations.  Therefore, $k$ decides the trade-off between the load size of each iteration and the number of iterations. $\epsilon_3$ controls the stopping criterion (line 12) of the RA-S algorithm and it is correlated with $\epsilon_1$ and $\epsilon_2$ according to Eq System \ref{eqs:1}. In general, a small $\epsilon_3$ implies a strict stopping criterion but reduces the number of RA sets required. Conversely, when $\epsilon_3$ is large, the stopping criterion become less strict but more RA sets are needed as $\epsilon_1$ and $\epsilon_2$ become smaller. 

\textbf{An optimization.} As shown in Alg. \ref{alg:1/2}, the Buchbinder's algorithm examines  the nodes one by one via a loop from line 3 to line 11, and the approximation ratio holds regardless of the examining order. Therefore, one may wonder that which examining order gives the best result. In this paper, we sort the nodes by $\pi(v)$ in a non-increasing order in advance. In particular, we first generate a bunch of RA sets $\mathcal{R}_l=\{R_1,...,R_l\}$ and use $\frac{\sum_{i=1}^{l}x(R_i,{v})}{l}$ as an estimate of $\pi(v)$ for sorting. We adopt such an examining order whenever the Buchbinder's algorithm is used. 

\textbf{Comparison with RA-T algorithm.} The major difference between RA-T and RA-S is that they may use different numbers of RA sets. Although RA-T and RA-S algorithms have the same asymptotic running time, the RA-S algorithm can be more efficient in many cases. Recall that $l$ is the number of RA sets used for estimating. In the RA-T algorithm, even if $\delta_2$ is much smaller than $\delta_1$, $l$ has to be set as $\max(\delta_1, \delta_2) \approx \delta_1$ so that both the $F(\mathcal{R}_l,V^{*})$ and $F(\mathcal{R}_l,V_{opt})$ are sufficiently accurate. In the RA-S algorithm, as shown in line 12 in Alg. \ref{alg:reverse_2}, we utilize Monte Carlo simulation to check whether $F(\mathcal{R}_l^*,V^{*})$ is accurate enough, and therefore $F(\mathcal{R}_l^*,V_{opt})$ can be a good estimate as long as $l$ is larger than $\delta_2^*$. As a result, if the stopping criterion is satisfied when $l$ is much smaller than $\delta_1^*$, RA-S returns a good solution without generating many RA sets. However, as observed in the experiments, RA-S does not dominate RA-T. 

\textbf{Comparison with SPM and RPM.} According to the analysis in Sec. \ref{sec:front} and this section, RA-T and RA-S are theoretically efficient than SPM and RPM with respect to running time, which is also confirmed in experiments shown later in Sec. \ref{sec:exp}. However, SPM and RPM can be very fast under parallel computing systems as executing simulations and generating realizations are highly parallelizable. Furthermore, the forward sampling method designed in Sec. \ref{sec:front} is potentially applicable to other USM problems beyond the PM problem, while RA-T and RA-S are proposed especially for the PM problem under the triggering model.

\section{Experiments}
\label{sec:exp}
In this section, we present the experiments for evaluating the proposed algorithms. Besides showing the performance of the algorithms proposed in this paper, the results herein also provide useful observations on the property of the PM problem.

\subsection{Experimental Settings}
Our experiments are performed on a server with 16 GB ram and a 3.6 GHz quadcore processor running 64-bit JAVA VM 1.6.

\begin{table}[t]
\centering
{\begin{tabular}{ p{1.6cm} || p{1.4 cm} || p{1.4 cm} || p{2.45 cm} } 
\textbf{Dataset}& \textbf{\#Edges}& \textbf{\#Nodes}& \textbf{Avg Degree}\\
\hline
Wiki& 104K& 7K& 14.7\\

HepTh& 353K& 28K& 12.6\\

HepPh& 422K& 35K& 12.1\\

Youtube& 6.0M& 1.1M& 5.4\\

Pokec& 30.6M& 1.6M& 19.1\\
\end{tabular}}
\caption{\textbf{Datasets.} }
\label{table:datasets}
\end{table}

\subsubsection{Datasets}
The datasets selected for experiments are listed in Table \ref{table:datasets}. All the datasets are borrowed from Leskovec \textit{et al.} \cite{snapnets}. \textbf{Wiki} is a who-votes-on-whom network collected from Wikipedia. \textbf{HepTh} and \textbf{HepPh} are arXiv citation graphs from the categories of high energy physics theory and high energy physics phenomenology, respectively. \textbf{Youtube} is a social network drawn from the Youtube video-sharing website. \textbf{Pokec} is a popular online social network in Slovakia. The original data of Youtube is undirected and we take each edge as a bi-directed edge to make it directed.

\subsubsection{Models.} As mentioned earlier, we adopt IC and LT models for experiments. For IC model, we consider two probability settings. One is constant probability (CP) setting where each edge has the propagation probability of 0.01 and the other one is the weighted cascade (WC) setting where $p_{(u,v)}=1/|N_v^-|$. To distinguish those two settings, we use IC-CP to denote the former and IC-WC for the later. For LT model, following the convention \cite{kempe2003maximizing, tang2015influence, tang2016profit}, the weight of edge $(u,v)$ is set as $1/|N_v^-|$. The price $P$, coupon $C$ and intrinsics are normalized within [0,1], shown as follows. $P$ is selected from $[0.2, 0.6]$ and $C$ is set as the 90\% of $P$. After setting $P$ and $C$, the intrinsic of each user is randomly generated from $[P-C,1]$ in uniform. Under this setting, every node can be an adopter once selected as a seed user.

\subsubsection{Algorithms}The algorithms considered in our experiments are shown as follows. We always set $\epsilon$ and $N $ as 0.4 and $n$, respectively.

\textbf{RA-S.} This is the algorithm proposed in Sec. \ref{subsec:RA-S}. For small graphs (Wiki, HepTh and HepPh) , we set that $k=5$ and $\epsilon_3=0.1$, and for large graphs (Youtube and Pokec) we set that $k=8$ and $\epsilon_3=0.3$.  Other parameters are calculated by solving Eq System \ref{eqs:1}. Furthermore, we immediately return the seed set if the $F(\mathcal{R}_l^*,V^{*})$ calculated in two consecutive iterations decreases by less than $2\%$ \footnote{As observed in the experiments, $F(\mathcal{R}_l^*,S)$ decreases with the increase of $l$. Such a phenomenon is also observed in \cite{arora2017debunking}.}. 

\textbf{RA-T.} This is the algorithm proposed in Sec. \ref{subsec:RA-T}. In experiments, we set the maximum of RA sets as 5,000,000. The results under this setting are also used later for discussing how many RA sets needed in obtaining good estimates.

\textbf{SPM}. This is the algorithm proposed in Sec. \ref{subsec:simu}. The parameter $l$ here controls the number of simulations used for each estimating. In the previous works regarding the IM problem, the number of simulations is typically set as 10,000 to obtain an accurate estimation. However, in this paper, we set $l$ as 2,000 for two reasons. First, it is reported in \cite{arora2017debunking} that 1,000 simulations are sufficient for small graphs. Second, SPM takes more than 24 hours on Wiki even with $l=2,000$, which means it is not a good choice even if it could provide better performance with larger $l$.

\textbf{RPM}. This is the algorithm proposed in Sec. \ref{subsec:realization}. The number of realizations is set as 1,000 to make it terminate in a reasonable time. SPM and RPM are tested only on Wiki under IC-CP and LT as they are very time-consuming. 

\textbf{MaxInf.} As the first part of our objective function (Eq. (\ref{eq: objective})), the number of adopters is important in deciding the profit. Therefore, the algorithm that maximizes the number of adopters can be a reasonable heuristic for the PM problem. Note that maximizing the number of adopters is essentially the IM problem. We adopt one of the state-of-the-art IM algorithms, D-SSA \cite{nguyen2016stop}, as a heuristic algorithm for the PM problem. In particular, we set the size of the seed set to be $n*i/50$ for $i=1,...,50$. For each size, we run the D-SSA to obtain a seed set and evaluate the profit by simulation. Finally, we select the result with the highest profit. This approach is denoted as MaxInf.

\textbf{HighDegree.} Selecting the nodes with the highest degree is a popular heuristic for the maximization problems in social networks. We first randomly decide the size of seed nodes and then select the nodes with the highest degree as the seed set. For each graph, the above process is executed for 100 times. Finally, we select the result with the highest profit.

The profits under the seed sets produced by the above algorithms are finally evaluated by 10,000 simulations. On large graphs, we show the results of RA-T and a simplified version of MaxInf where, instead of enumerating it from  $n/50$ to $n$, the size of the seed set of MaxInf is set as the same as that of the result produced by RA-T.

\begin{figure*}[pt]
\centering
  
\subfloat[Wiki \& IC-CP]{
\label{fig:wiki_ic}
\includegraphics[width=0.20\textwidth]{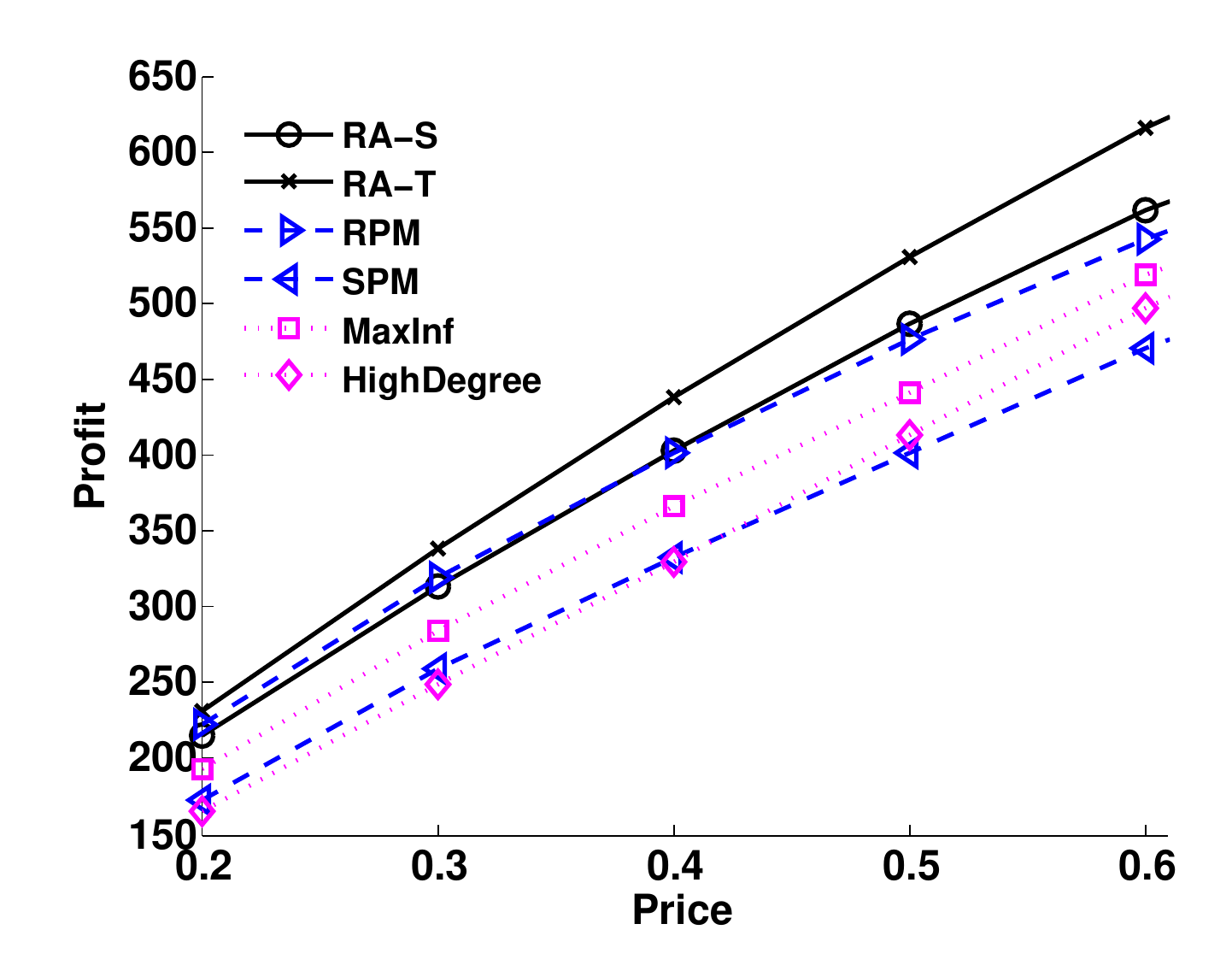}
\includegraphics[width=0.20\textwidth]{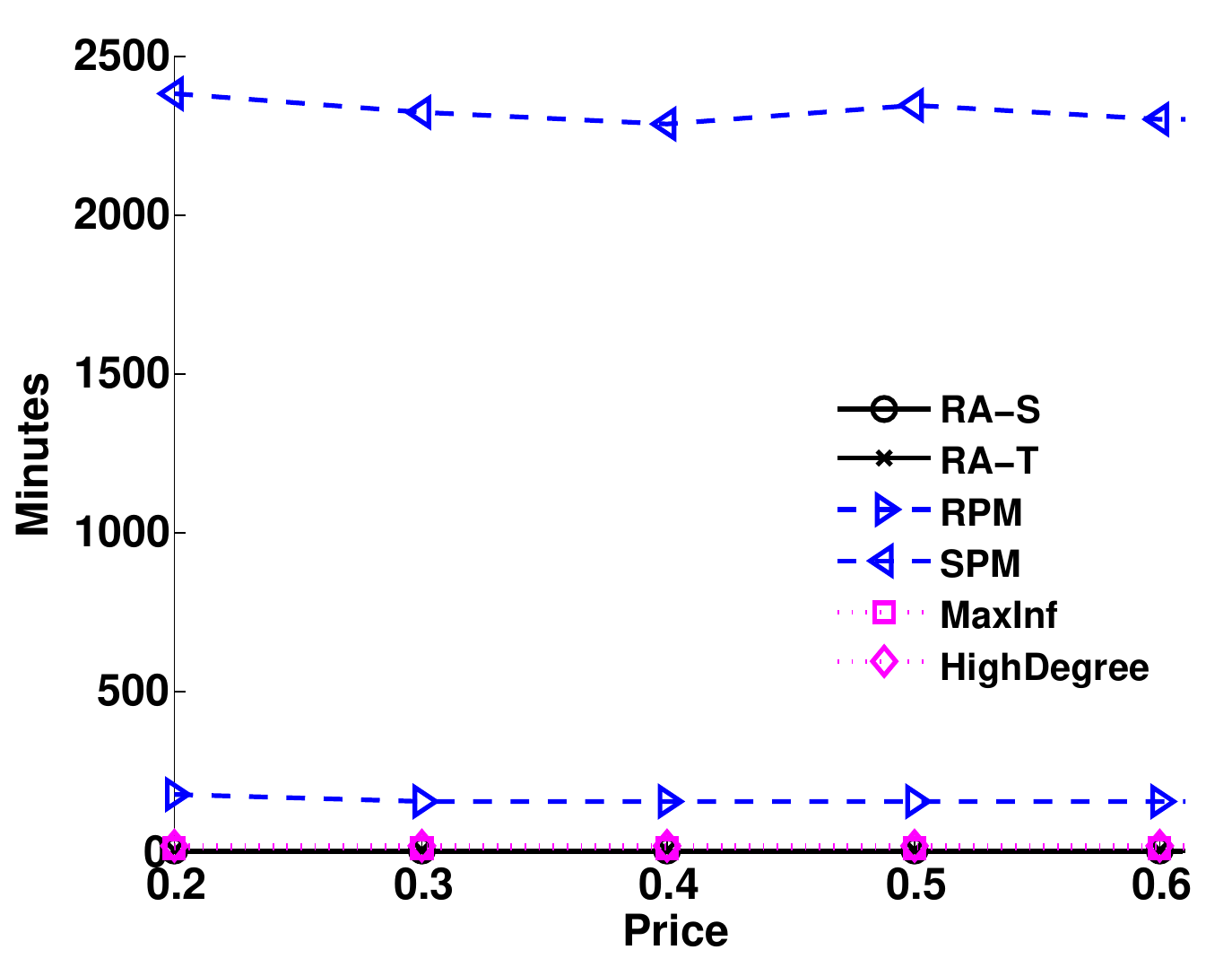}} \hspace{10mm} 
\subfloat[Wiki \& LT]{
\label{fig:wiki_lt}
\includegraphics[width=0.20\textwidth]{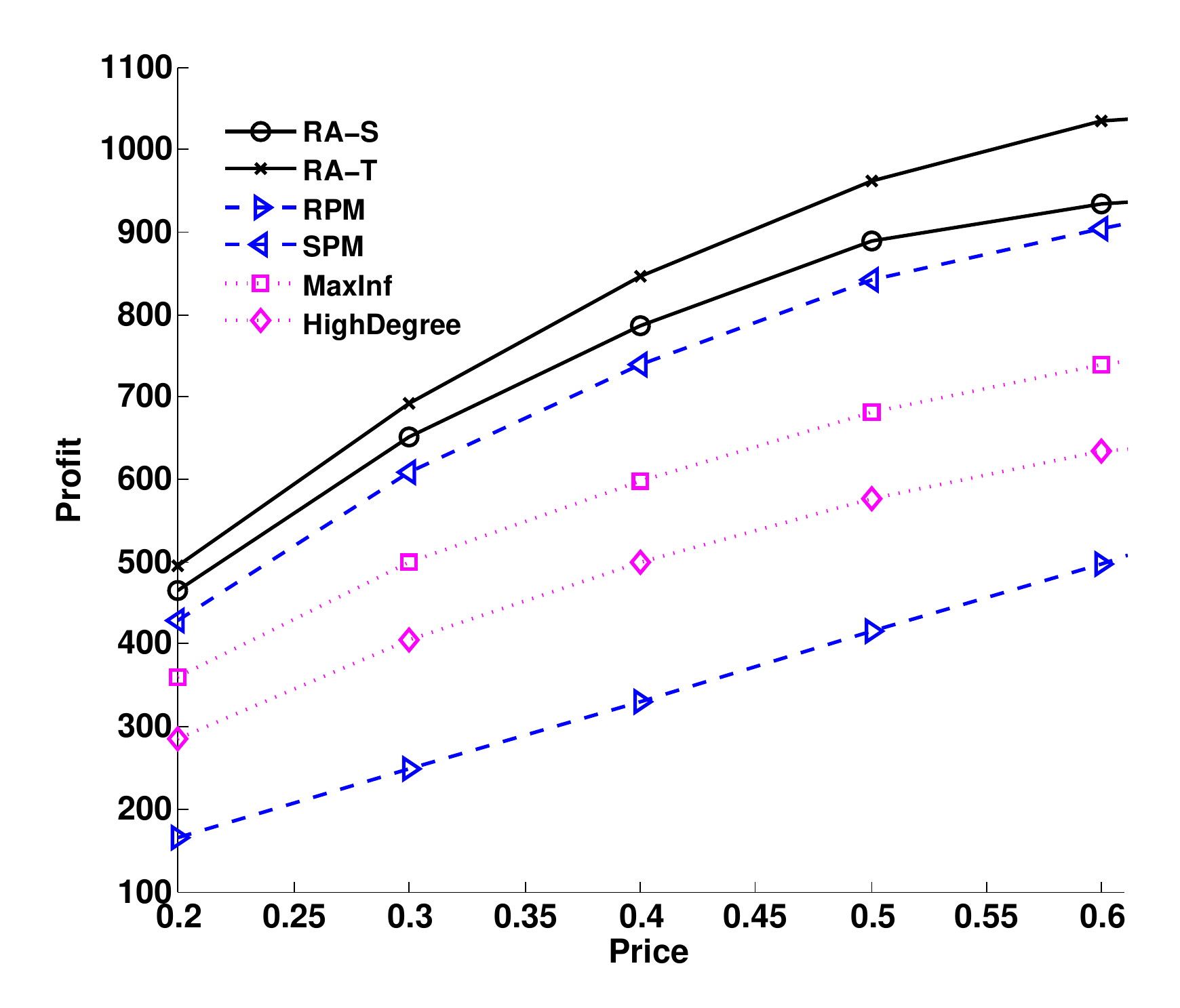}
\includegraphics[width=0.20\textwidth]{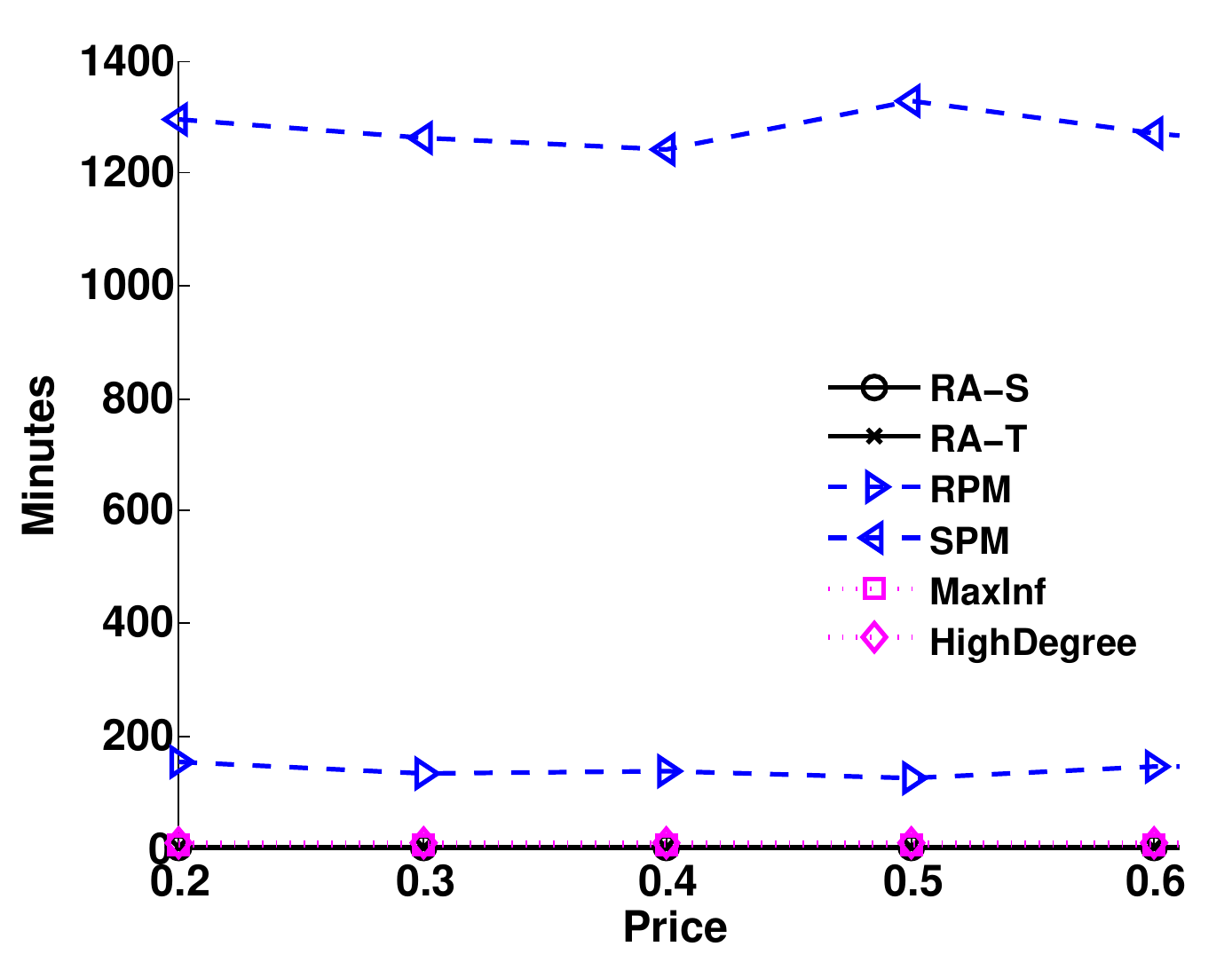}} \vspace{-1.5mm}

\subfloat[HepTh \& LT]{
\label{fig:hepth_lt}
\includegraphics[width=0.20\textwidth]{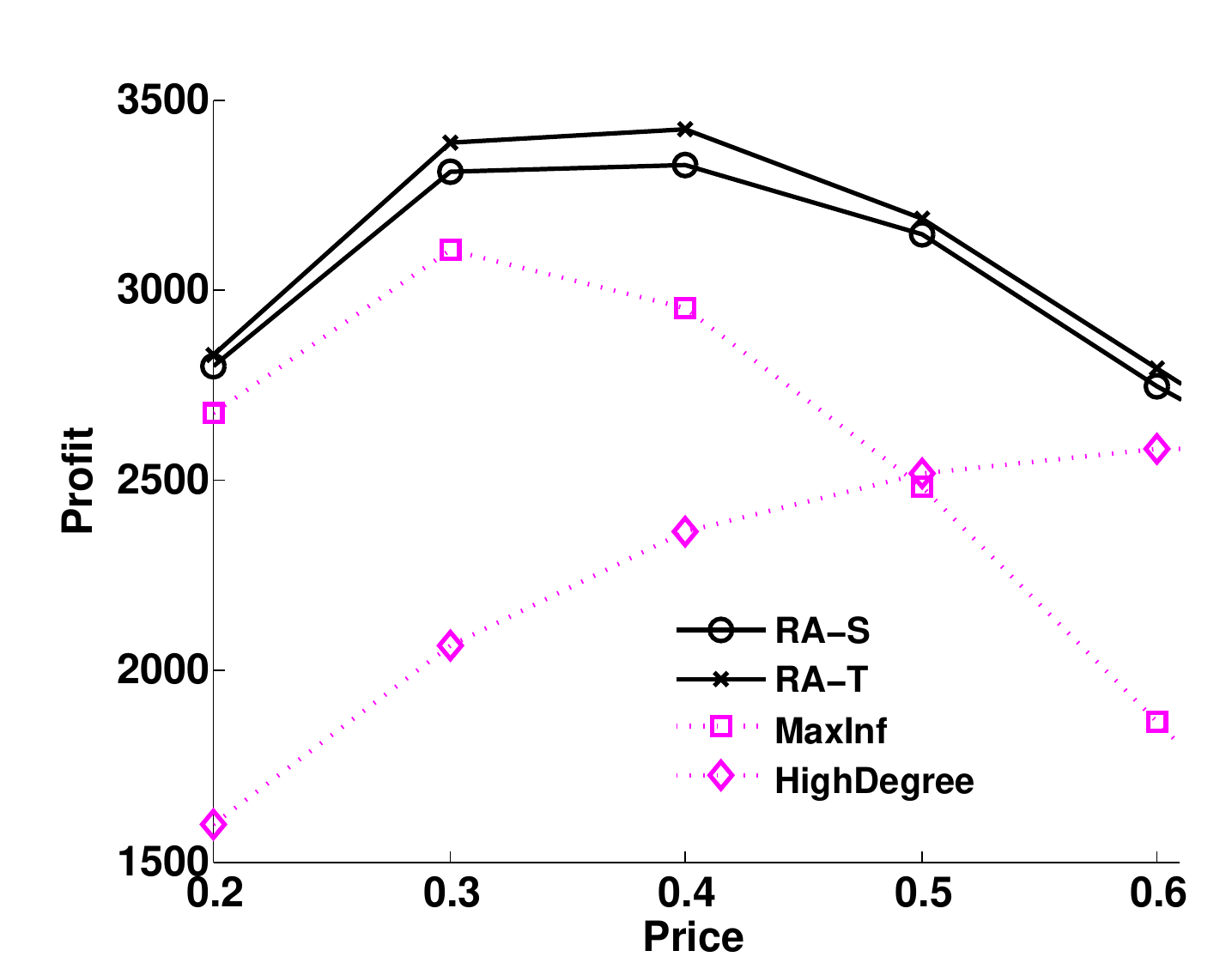}
\includegraphics[width=0.20\textwidth]{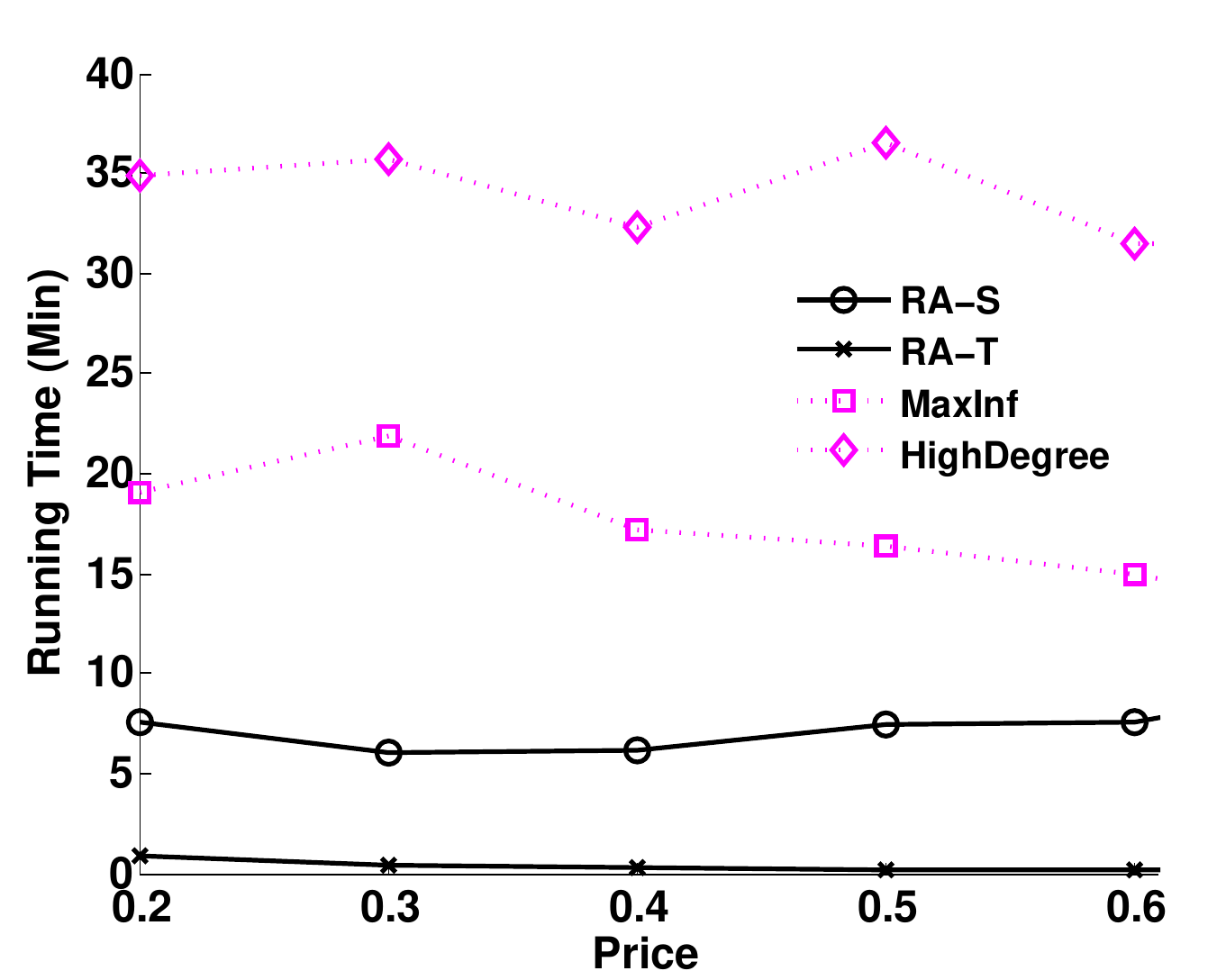}}\hspace{10mm}
\subfloat[HepTh \& IC-CP]{
\label{fig:hepth_ic}
\includegraphics[width=0.20\textwidth]{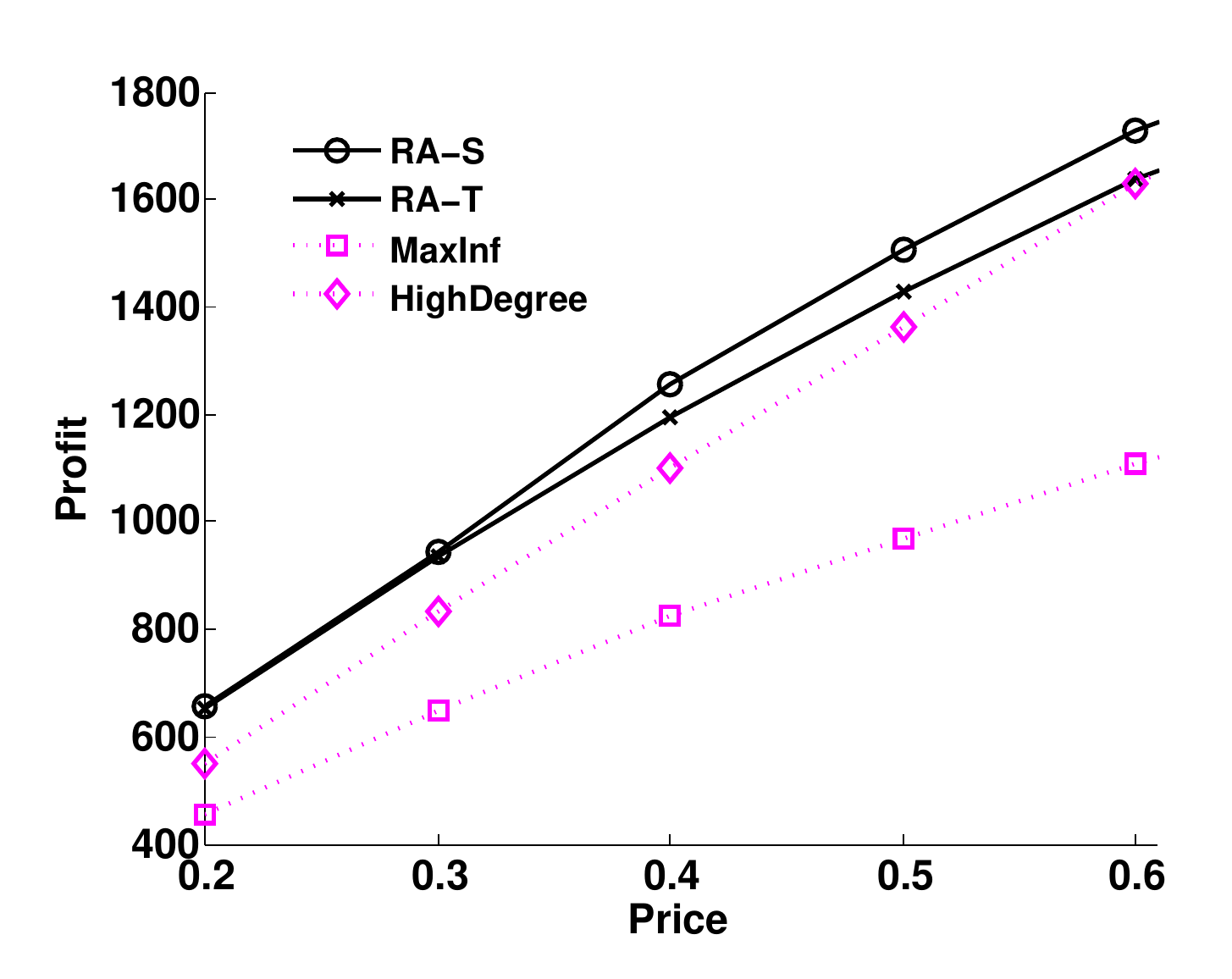}
\includegraphics[width=0.20\textwidth]{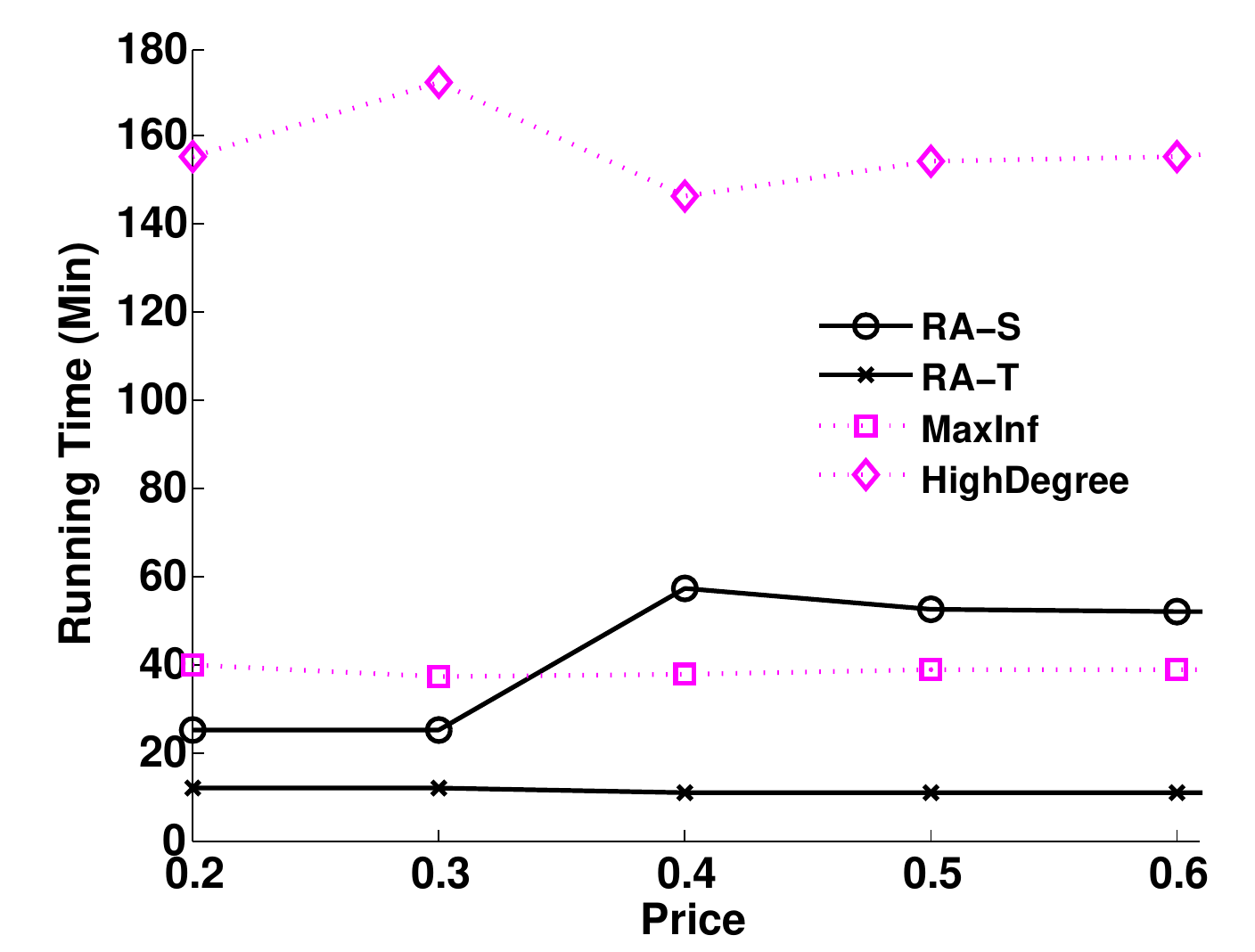}}\vspace{-1.5mm}

\subfloat[HepPh \& LT]{
\label{fig:hepph_lt}
\includegraphics[width=0.20\textwidth]{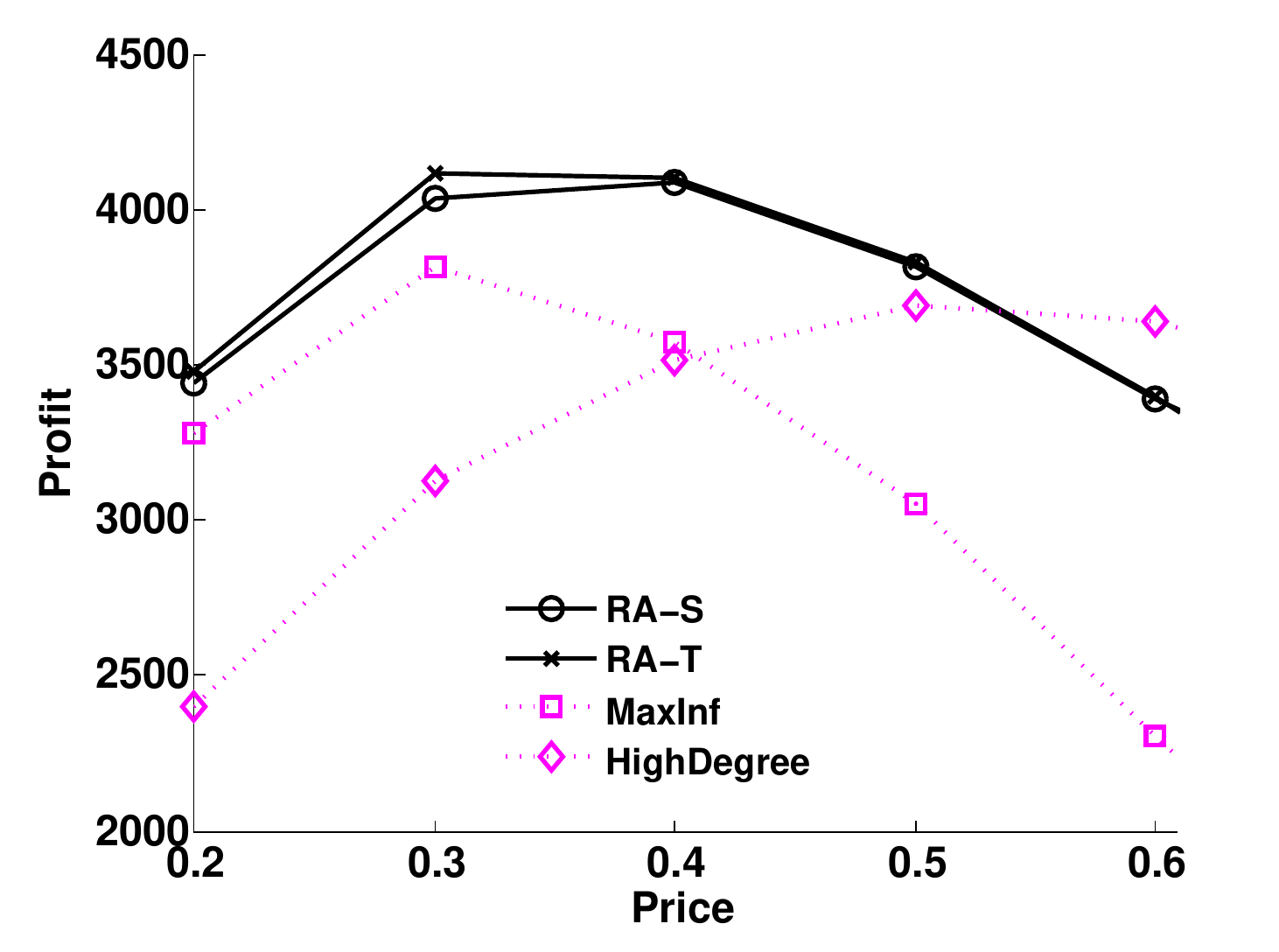}
\includegraphics[width=0.20\textwidth]{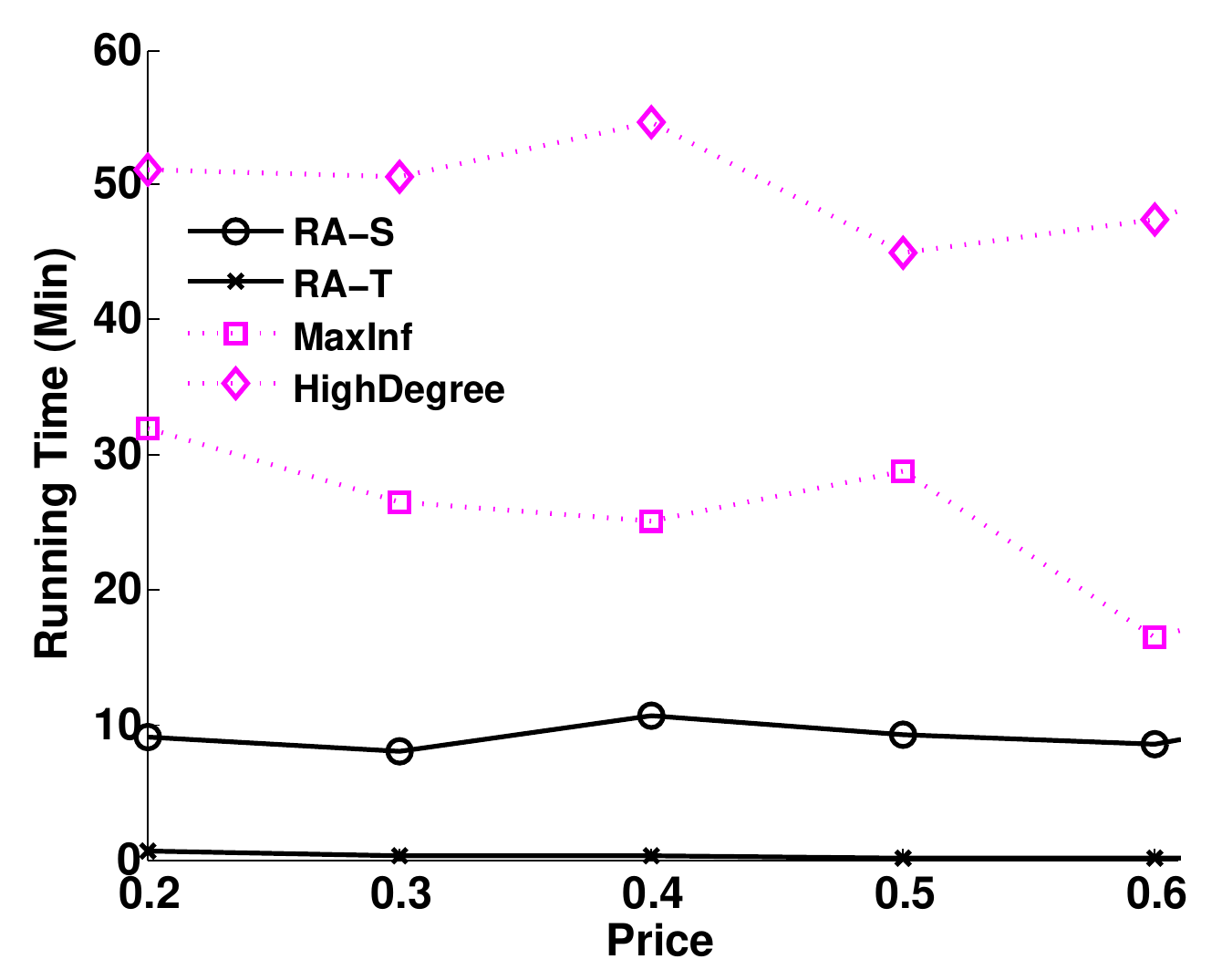}}\hspace{10mm} 
\subfloat[HepPh \& IC-CP]{
\label{fig:hepph_ic}
\includegraphics[width=0.20\textwidth]{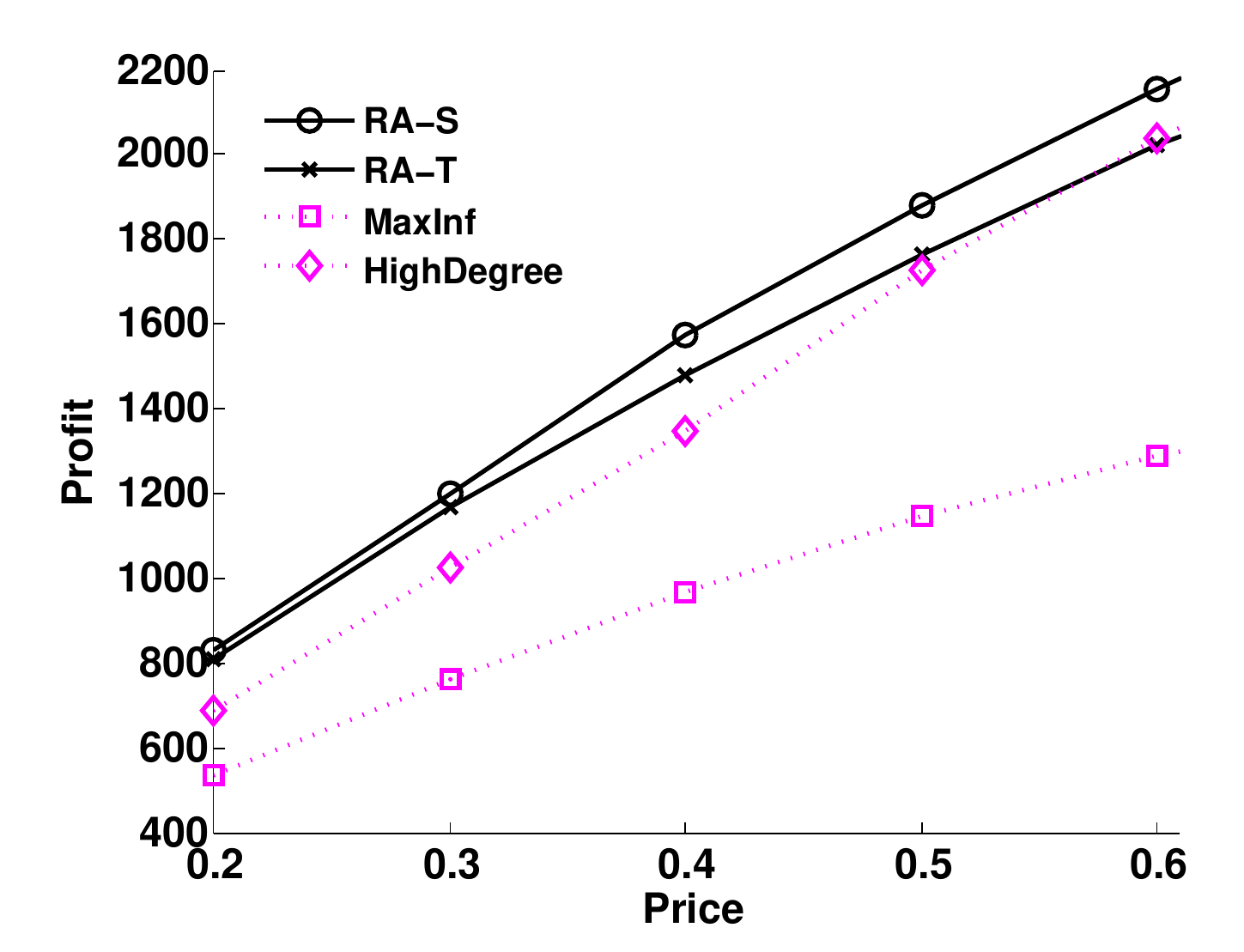}
\includegraphics[width=0.20\textwidth]{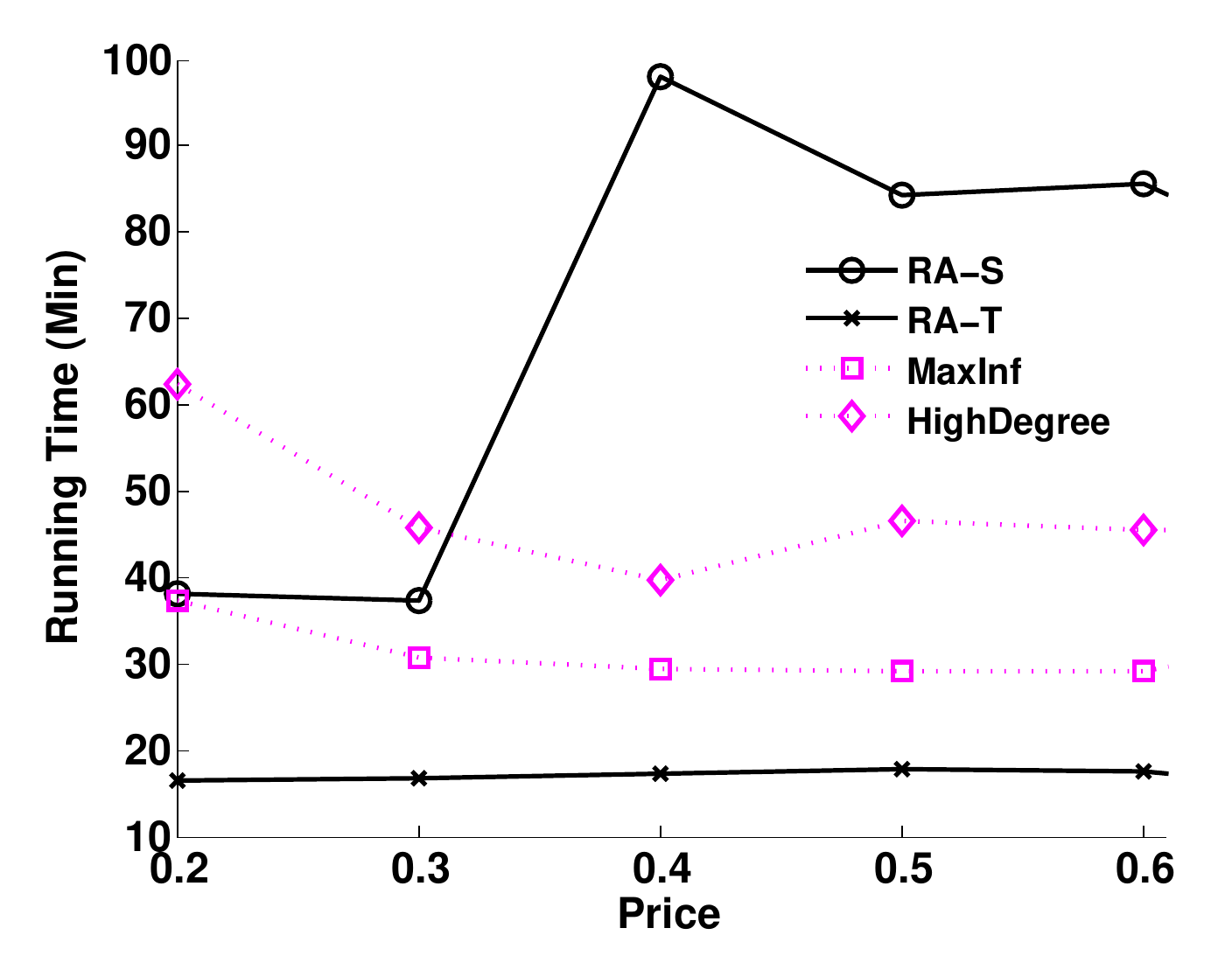}} \vspace{-1.5mm}  

\subfloat[Youtube \& WC]{
\label{fig:youtube_wc}
\includegraphics[width=0.20\textwidth]{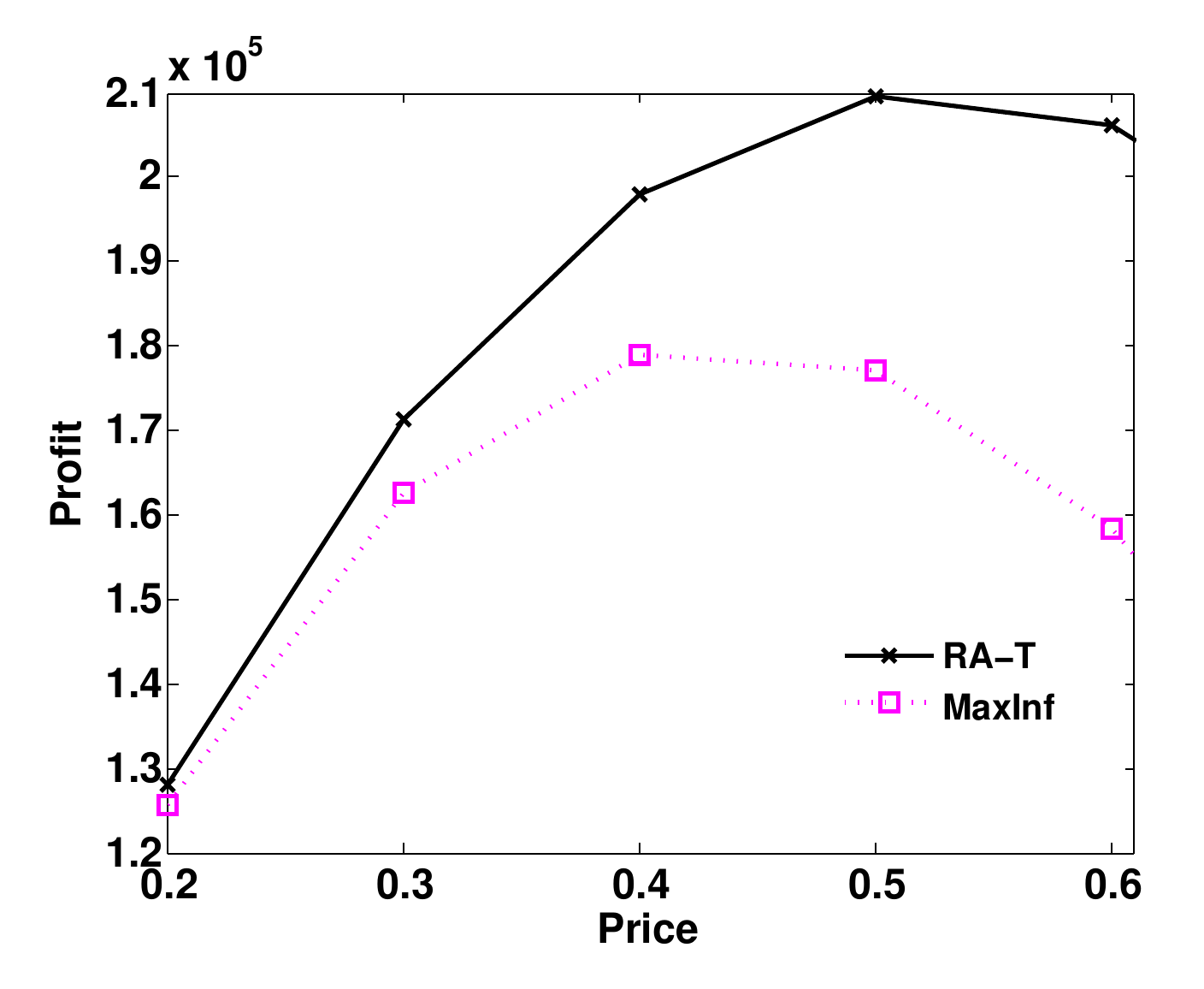}
\includegraphics[width=0.20\textwidth]{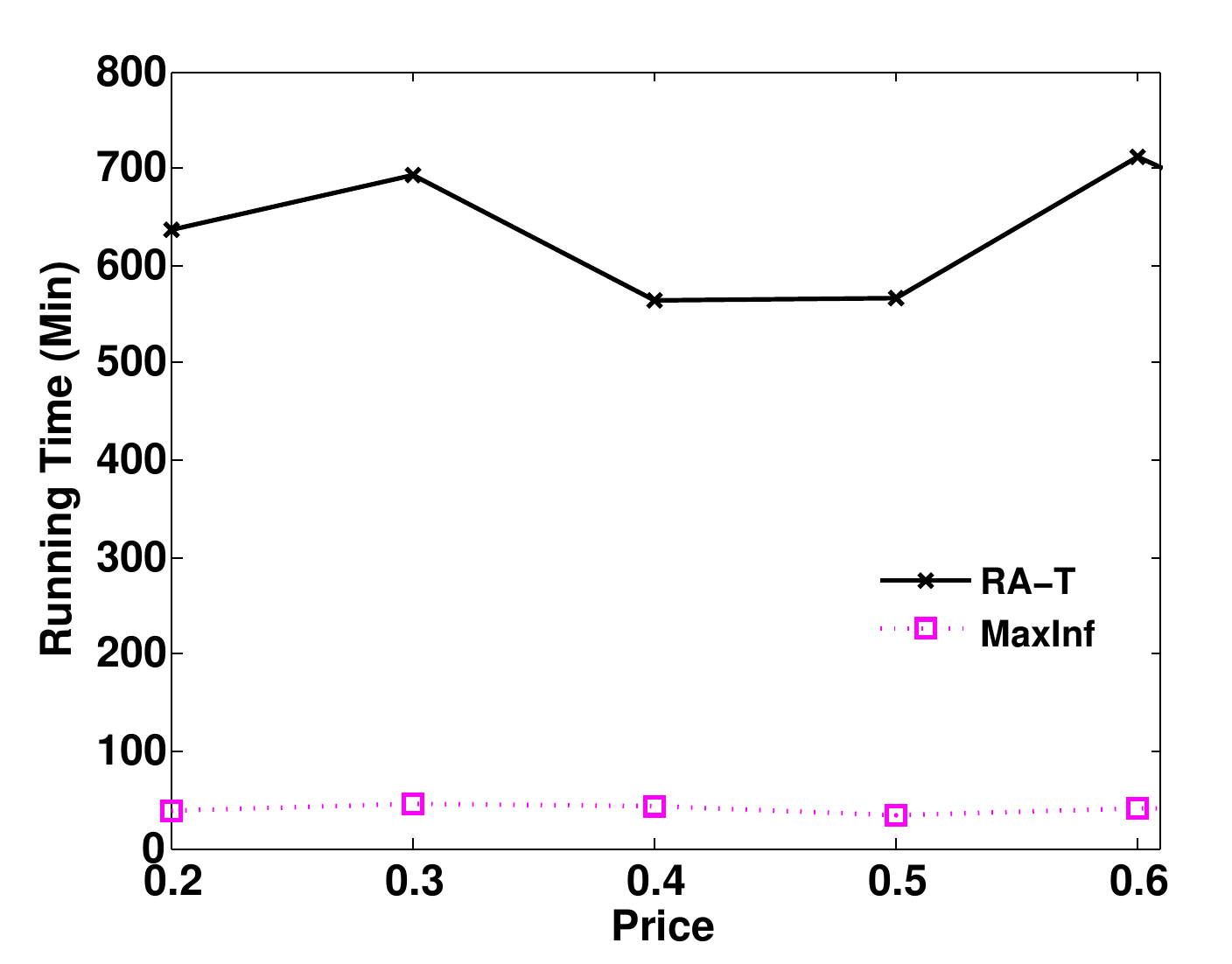}}\hspace{10mm} 
\subfloat[Youtube \& LT]{
\label{fig:youtube_lt}
\includegraphics[width=0.20\textwidth]{images/new/hepph_lt.pdf}
\includegraphics[width=0.20\textwidth]{images/time/hepph_lt_time.pdf}} \vspace{-1.5mm}

\subfloat[Pokec \& IC-CP]{
\label{fig:pokec_ic}
\includegraphics[width=0.20\textwidth]{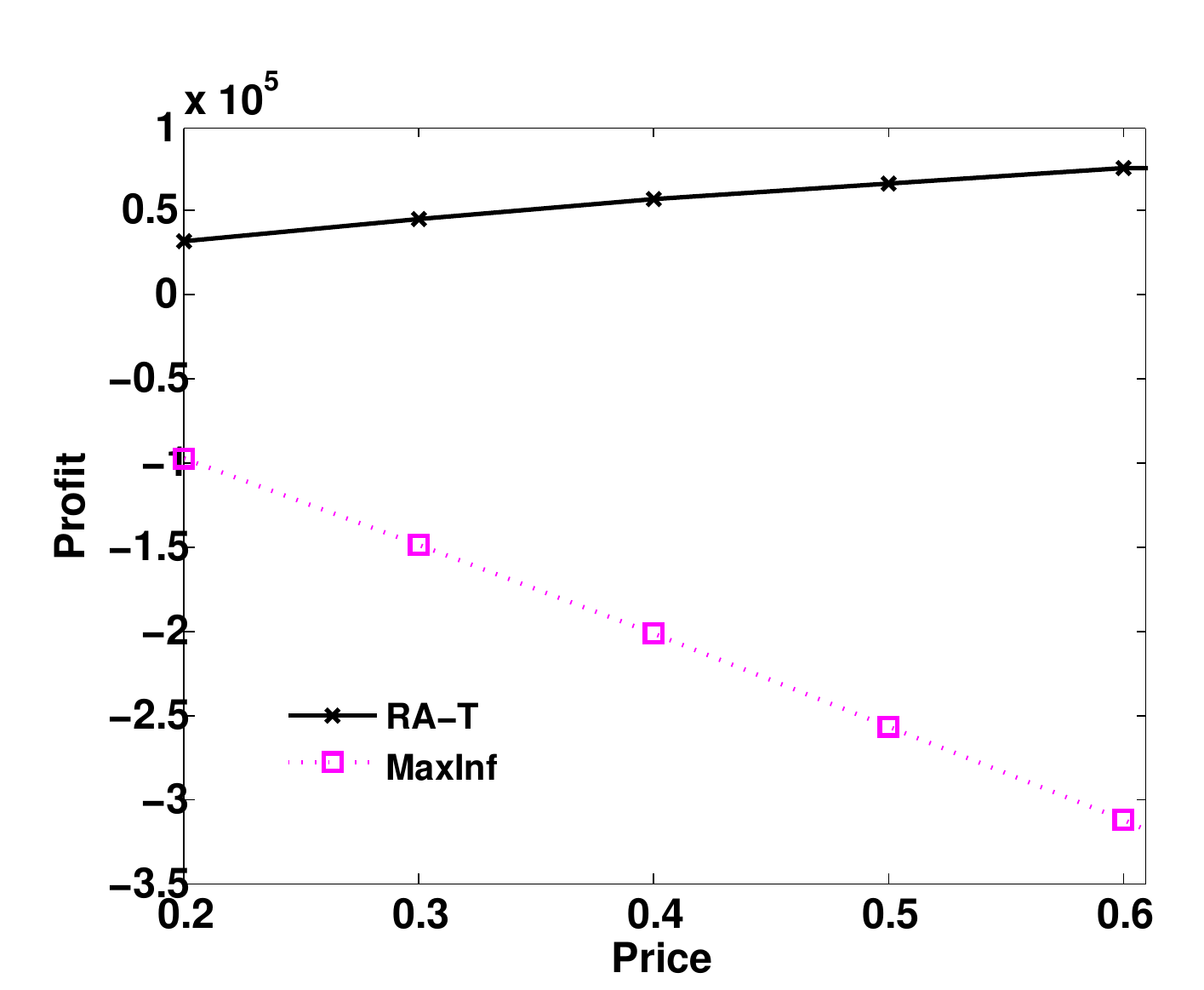}
\includegraphics[width=0.20\textwidth]{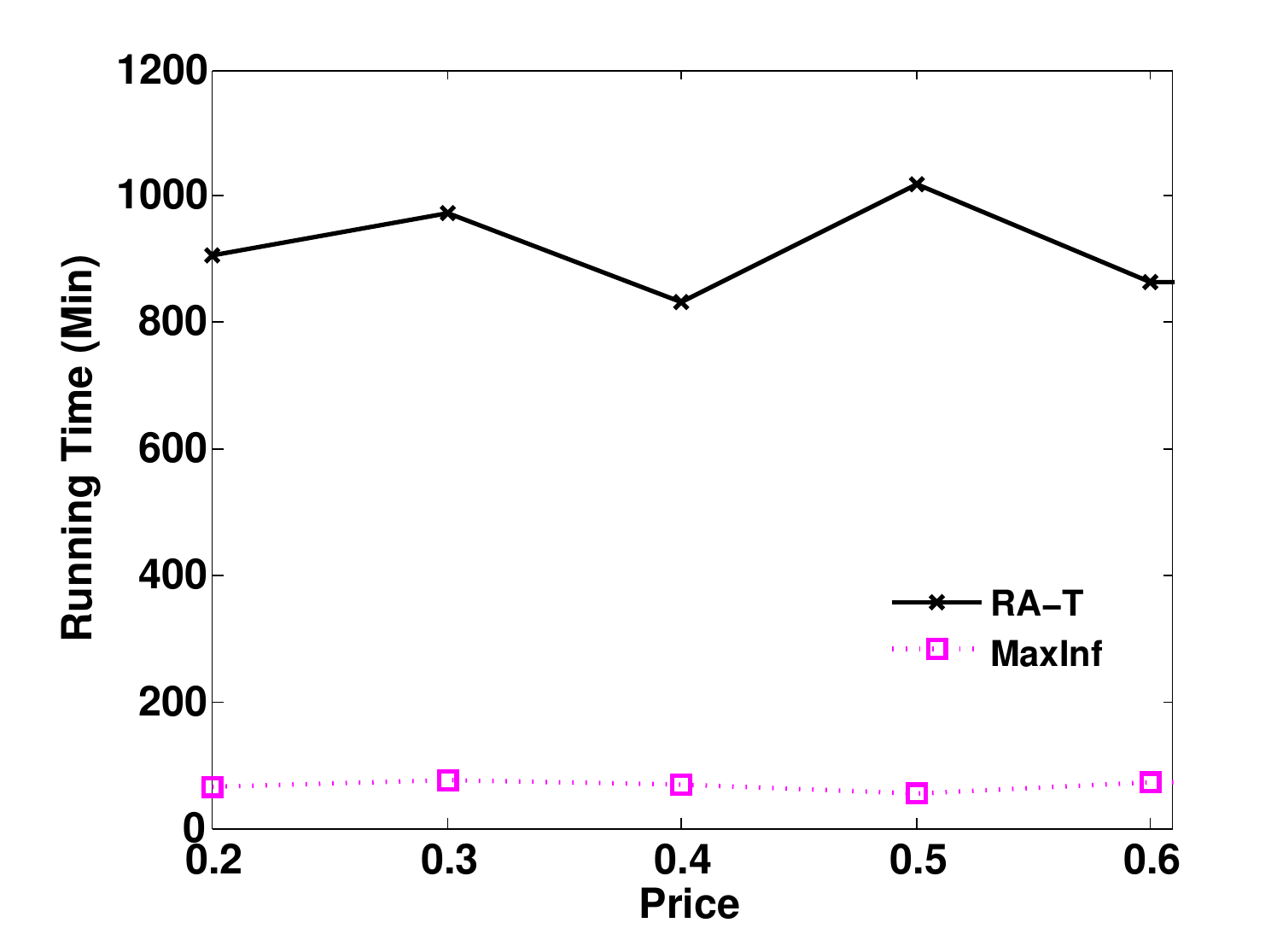}}\hspace{10mm}  
\subfloat[Pokec \& WC]{
\label{fig:pokec_wc}
\includegraphics[width=0.20\textwidth]{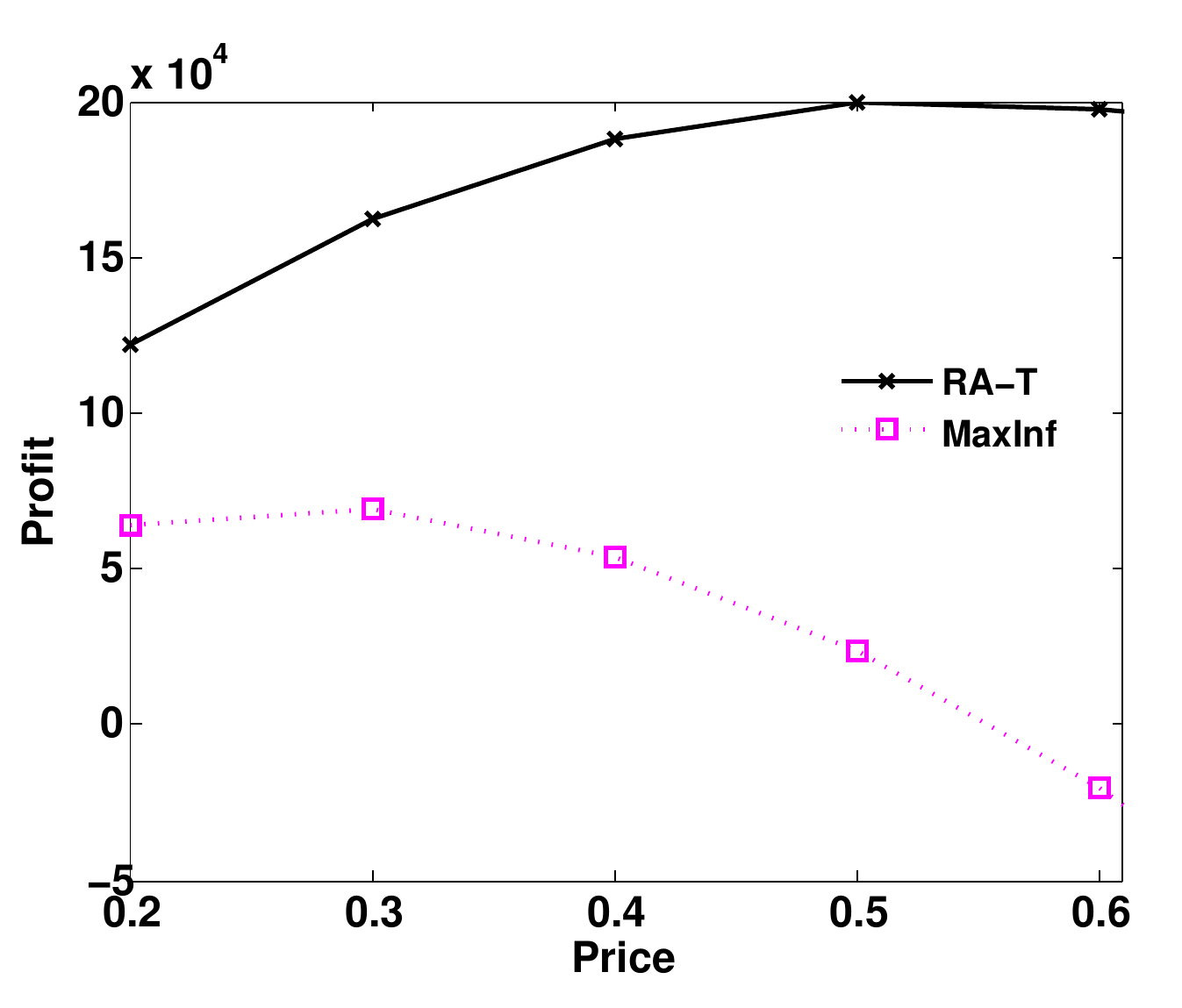}
\includegraphics[width=0.20\textwidth]{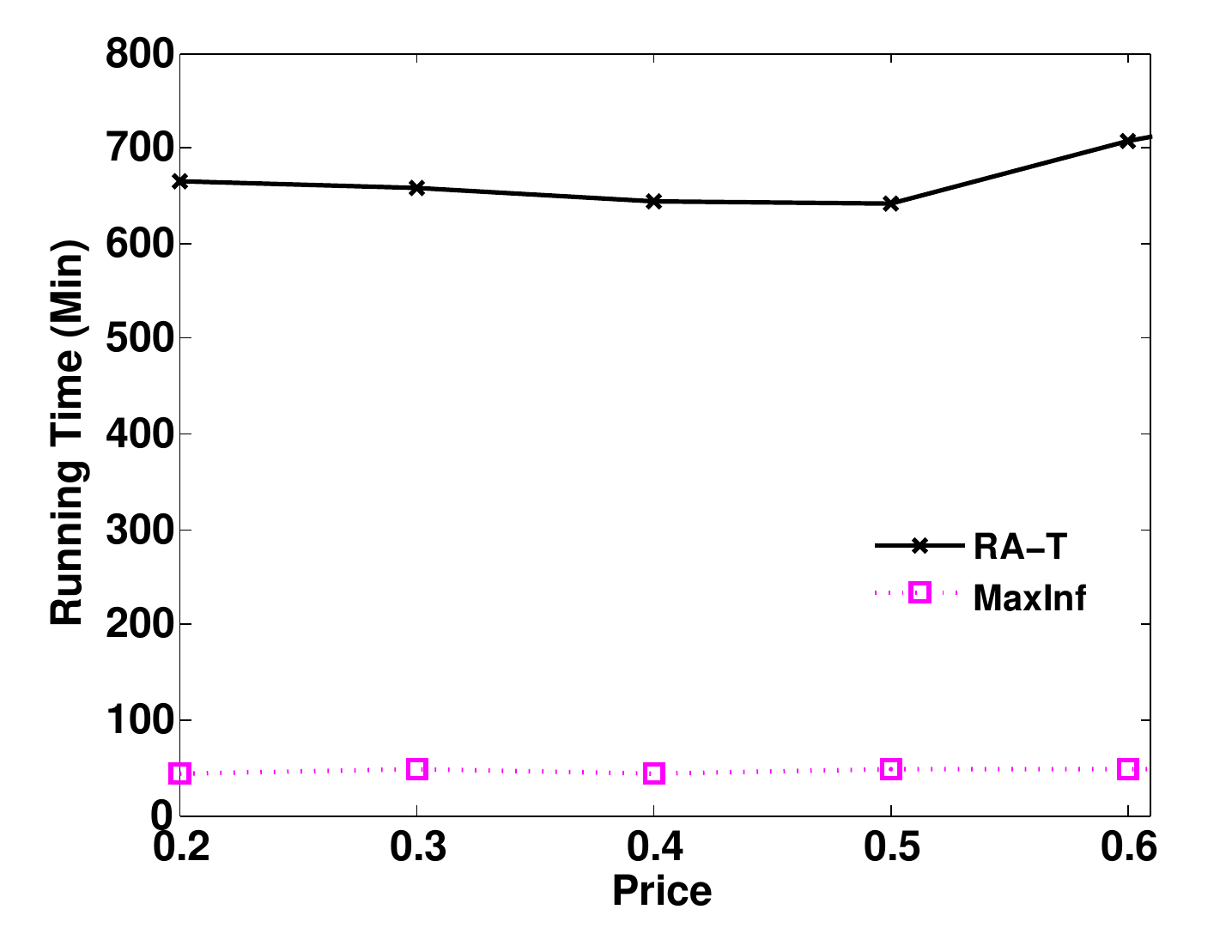}} \vspace{-1.5mm}

\caption{\small \textbf{Experimental results.} Each subset corresponds to one dataset under one model, where the left and right figures show the profit and running time ,respectively.} 
\label{fig:result}
\vspace{-3mm}
\end{figure*}

\subsection{Discussions}
\label{subsec:discuss}
The experimental results are shown Fig. \ref{fig:result}. According to the experimental results, we make the following observations.

\textbf{RA-S and RA-T give the best overall performance for maximizing the profits.} As shown in the figures, RA-S and RA-T consistently give the highest profit on the considered datasets under different settings. While SPM and RPM are also approximation algorithms, they cannot always produce high-quality seed sets. As shown in Figs. \ref{fig:wiki_ic} and \ref{fig:wiki_lt}, with $l=1,000$ on graph Wiki, RPM performs well under IC-CP model but considerably worse under LT model. Note that the gaps between the lines of different algorithms are not very significant because the price of the product is normalized within $[0, 1]$. 

\textbf{MaxInf and HighDegree are sometimes comparable to RA-S and RA-T.} As shown in Fig. \ref{fig:result}, MaxInf performs relatively good when the price is very low. This is because when $P$ is small almost all the nodes can be adopters once activated due to the setting of the intrinsic, and therefore the first part $\pi()$ dominates the objective function $f()$.  Since MaxInf aims to maximize $\pi()$, it has a good performance in this case. When the price is very high, except the seed users, most of the users cannot be adopters and therefore the objective function is almost monotone, which makes HighDegree accidentally comparable to RA-S e.g. Fig. \ref{fig:hepth_lt}. As an extreme case, when $P$ approaches to 1, the seed users cannot activate any neighbor and thus the objective function is strictly monotone. In this case, $f(S)=(P-C)|S|$ only depends on the size of the seed set and the considered problem becomes less interesting.

\textbf{RA-S and RA-T are efficient with respect to running time.} As shown in the graphs, RA-S and RA-T are even more efficient than the two heuristics, MaxInf and HighDegree, under certain circumstances. This is because the two heuristics still need to utilize Monte Carlo simulations to compare the profit under different seed sets. SPM and RPM are not scalable to large graphs as they also employ Monte Carlo simulations. For example, as shown in Figs \ref{fig:wiki_ic} and \ref{fig:wiki_lt}, on graph Wiki, SPM takes 38.6 and 21.2 hours under IC-IC and LT models, respectively.

\textbf{Comparison between SPM and RPM.} As shown in Fig. \ref{fig:wiki_ic}, comparing with SPM, RPM results in a higher profit and also takes much less time on graph Wiki under IC-CP model, which gives the evidence that RPM is sometimes superior to SPM. However, RPM requires to load the realizations into memory and therefore it consumes more memory than SPM does.

\textbf{How many RA sets needed to produce a good seed set?} As shown in Sec. \ref{sec:reverse}, a great deal of effort has been devoted to computing the number of RA sets needed to produce a good seed set for maximizing the profit. This is actually the key problem for many influence related problems where the reverse sampling technique is used, as shown in  \cite{tang2014influence,tang2015influence,tong2017efficient,lu2015competition}.  While more RA sets bring a higher quality of the seed sets, the quality increases very slow when a large number RA sets have already been used, and therefore one should stop generating RA sets if the current seed set is satisfactory. We observe that the algorithms proposed in this paper can sometimes overly generate RA sets. For example, on graph HepPh under the LT model, as shown in Fig. \ref{fig:hepph_lt}, RA-S and RA-T have the similar performance but RA-S is ten times slower than RA-T, which means the number of RA sets used in RA-S is larger than that of used in RA-T. In fact, we believe there is room for improvement of the algorithms proposed in this paper. 

\textbf{Profit V.s. Price.} As observed in the experiments, the profit does not necessarily increase with the increase of price, especially under the LT model as shown in Figs. \ref{fig:hepth_lt} and \ref{fig:hepph_lt}. Such a phenomenon coincides with the setting of the intrinsic, because a large $P$ narrows the set of users who are able to be adopters and therefore leads a decrease in total profit.

\textbf{Comparing to IM problem.} It is worthy to note that, in general, the PM problem considered in this paper takes more time to solve than the IM problem does. Although the SPM algorithm proposed in this paper and the CELF algorithm \cite{leskovec2007cost} for IM problem are both designed with the idea of forward sampling, SPM is not able to run on HepPh even with $l=2,000$ while CELF terminates within 40 hours on the same graph as shown in \cite{arora2017debunking}. There are several possible explanations for this phenomenon. First, the PM itself is an unconstrained optimization problem while the IM problem is a budgeted maximization problem, and thus PM requires more effort for sampling because $\pi(V_{opt})$ can be very large. Second, the lazy forward sampling method \cite{leskovec2007cost} significantly reduces the running time of the greedy algorithm for solving the IM problem without loss of the approximation guarantee. Unfortunately, currently we are not aware of any technique that can be used to boost the efficiency of the Buchbinder's algorithm (Alg. \ref{alg:1/2}).

\textbf{On large graphs.} The results on Youtube and Pokec are shown in Figs. \ref{fig:youtube_lt}, \ref{fig:youtube_wc}, \ref{fig:pokec_ic} and \ref{fig:pokec_wc}. We do not show the experimental results of HighDegree because it still performs worse than RA-T does. The results of RA-S are not included because it runs out of memory on those graphs. Note that our experiments are done on a machine with 16 GB ram and we observe that RA-S is able to run on large graphs if the size of ram is increased to 32 GB or more. According to the results, even though the number of RA sets is limited to 5,000,000, RA-T still performs much better than MaxInf does.  On large graphs, MaxInf takes much less time than RA-T does, because the size of the seed set in MaxInf is given by the result of RA-T and consequently MaxInf no longer has to perform Monte Carlo simulations to compare the profits under the seed sets with different size. As shown in Figs.  \ref{fig:pokec_ic} and \ref{fig:pokec_wc}, MaxInf can sometimes result in negative profit, which again indicates that maximizing the influence may not produce a high profit. As shown in Figs \ref{fig:youtube_lt}, \ref{fig:youtube_wc}, RA-T is able to terminate within 10 hours on Youtube and Pokec.

\section{Conclusion and Future Work}
\label{sec:conclustion}
In this paper, we consider the problem of maximizing the profit with the usage of the coupon in social networks. By integrating coupon system into the triggering spreading model, we formulate the PM problem as a combinatorial optimization problem. Based on our formulation, several approximation algorithms are provided together with their analysis. The performance of the proposed methods is evaluated by experiments. The analysis of the PM problem is actually applicable to other influence-related problems which are non-monotone submodular, and therefore it provides an algorithmic framework. 

The major future work of this paper is derived from the observations of the experiments. As discussed in Sec. \ref{subsec:discuss}, to improve the efficiency of RA-T and RA-S, one has to consider the number of RA sets generated for estimating. On one hand we need sufficient RA sets for an accurate estimate but, on the other hand, we should avoid generating too many RA sets to reduce the running time. Essentially, it asks for a good stop criterion for generating RA sets in such kind of algorithms like RA-S and RA-T. Second, any technique that can boost the efficiency of the Buchbinder's algorithm can definitely reduce the running time of RA-S and RA-T. Thus, finding such techniques is a promising future work. Finally, we note that the heuristics designed based on RA sets may effectively solve the PM problem while taking much less time. We leave this as part of our future work.

\ifCLASSOPTIONcompsoc
  \section*{Acknowledgments}
  
\else
  \section*{Acknowledgment}
\fi

This work was supported by National Natural Science Foundation of China under 61472272.

\ifCLASSOPTIONcaptionsoff
  \newpage
\fi

\bibliographystyle{IEEEtran}
\bibliography{sigproc_short}

%

\begin{IEEEbiographynophoto}{Guang-mo Tong} is a Ph.D. candidate in the Department of Computer Science at the University of Texas at Dallas. He received his BS degree in Mathematics and Applied Mathematics from Beijing Institute of Technology in July 2013. His research interests include computational social systems, data communication and real-time systems. He is a student member of the IEEE.
\end{IEEEbiographynophoto}

\begin{IEEEbiographynophoto}{Weili Wu}(M’00) is currently a Full Professor with the Department of Computer Science, the University of Texas at Dallas, Dallas, TX, USA. She received the Ph.D. and M.S. degrees from the Department of Computer Science, University of Minnesota, Minneapolis, MN, USA, in 2002 and 1998, respectively. Her research mainly deals with the general research area of data communication and data management. Her research focuses on the design and analysis of algorithms for optimization problems that occur in wireless networking environments and various database systems.
\end{IEEEbiographynophoto}

\begin{IEEEbiographynophoto}{Ding-Zhu Du} received the M.S. degree from the Chinese Academy of Sciences in 1982 and the Ph.D. degree from the University of California at Santa Barbara in 1985, under the supervision of Professor Ronald V. Book. Before settling at the University of Texas at Dallas, he worked as a professor in the Department of Computer Science and Engineering, University of Minnesota. He also worked at the Mathematical Sciences Research Institute, Berkeley, for one year, in the Department of Mathematics, Massachusetts Institute of Technology, for one year, and in the Department of Computer Science, Princeton University, for one and a half years. He is the editor-in-chief of the Journal of Combinatorial Optimization and is also on the editorial boards for several other journals. Forty Ph.D. students have graduated under his supervision. He is a member of the IEEE
\end{IEEEbiographynophoto}

\appendices

\section{Chernoff Bound and Martingale}
\label{appendix:1}
Let $Z_i \in [0,1]$ be $l$ i.i.d random variables where $E(Z_i)=\mu$. The Chernoff bounds \cite{motwani2010randomized} states that 
\begin{equation}
\label{eq:chernoff_1}
\Pr\Big[\sum Z_i - l \cdot \mu \geq \delta \cdot l \cdot  \mu  \Big]\leq \exp(-\frac{l \cdot \mu \cdot \delta^2}{2+\delta}),
\end{equation}
and
\begin{equation}
\label{eq:chernoff_2}
\Pr\Big[\sum Z_i - l \cdot \mu \leq -\delta \cdot l \cdot \mu \Big]\leq \exp(-\frac{l \cdot \mu \cdot \delta^2}{2}),
\end{equation}  
for $\delta > 0$.

Martingale method provides good error estimates for a sequence of variables where the independence assumptions is not required. It has been proved by Tang \textit{et al.} \cite{tang2016profit} that, for any seed set $S$ and a sequence of RA sets $R_1,...,R_l$, the random variables $x(R_i,l)$, for $1 \leq i\leq l$, are martingales and they satisfy the following concentration inequalities, for any $\delta \geq 0$,
\begin{equation}
\label{eq:maringale_1}
\Pr[\sum x(R_i,S)-l\cdot \mu\geq \delta \cdot l\cdot \mu]\leq \exp(-\frac{l\cdot  \mu\cdot \delta^2}{2+\frac{2}{3}\delta})
\end{equation} 
and
\begin{equation}
\label{eq:maringale_2}
\Pr[\sum x(R_i,S)-l\cdot \mu \leq -\delta \cdot l\cdot \mu]\leq \exp(-\frac{l\cdot \mu\cdot  \delta^2}{2}),
\end{equation} 
and
where $\mu=\mathbb{E}[ x(R_i,l)]=\frac{\pi(S)}{n}$.
\section{Proofs}
\label{appendix:proofs}

\subsection{Proof of Theorem \ref{theorem:accuracy}}
\label{sebsec:proof_theorem:accuracy}

Recall that, in the $i$-th iteration of the loop, node $v_i$ is considered, and $X_i$ and $Y_i$ are the subsets created in line 9 or 11. Let $S^{
*}$ be the optimal solution and define that $S^{*}_i=(S^{*}\cup X_i)\cap Y_i$ for $1 \leq i \leq n$. Note that $S^{*}_0=S^{*}$ and $S^{*}_n=X_n=Y_n$. For a set $S$ and a node $v$, define that $S+v=S\cup \{v\}$ and $S-v=S\setminus \{v\}$. Because Alg. \ref{alg:1/2} is in fact a randomized algorithm, we consider the expected value of $h(X_n)$. The following is a key lemma to prove Theorem \ref{theorem:accuracy}.

\begin{lemma}
\label{lamma:appendix_2:lemma_2}
For each $1 \leq i \leq n$,
\begin{eqnarray}
&&\mathbb{E}[h(S^{*}_{i-1})-h(S^{*}_i)] \nonumber\\
&&\leq \frac{1}{2}\mathbb{E}[h(X_i)-h(X_{i-1})+h(Y_i)-h(Y_{i-1})]+\frac{2\epsilon }{n}\cdot h(S^{*}) \nonumber
\end{eqnarray}
\end{lemma}
\begin{proof}
It suffices to prove that it is true conditioned on every possible event by the $i$-th iteration. Now we fix $X_{i-1}$ and $Y_{i-1}$ and consider the $i$-th iteration. By Eq. (\ref{eq: submodular}), $h(X_{i-1}+v_i)-h(X_{i-1})+h(Y_{i-1}- v_i)-h(Y_{i-1})$ is always non-negative due to the submodularity of $h()$. Thus, 
\begin{eqnarray*}
&&\tilde{a}_i+\tilde{b}_i\\
&=&\tilde{h}(X_{i-1}+v_i)-\tilde{h}(X_{i-1})+\frac{2\epsilon }{n}\cdot L^{*}\\
&&+ \tilde{h}(Y_{i-1}-v_i)-\tilde{h}(Y_{i-1})+\frac{2\epsilon }{n}\cdot L^{*}\\
&&\{\text{By Eq. (\ref{eq:accuracy})}\}\\
&\geq& h(X_{i-1}+v_i)-h(X_{i-1})+h(Y_{i-1}-v_i)-h(Y_{i-1}) \\
&\geq& 0,
\end{eqnarray*}
and there are three cases to consider:

\textbf{Case 1.} ($\tilde{a}_i \geq 0$ and $\tilde{b}_i \leq 0$). In this case, $Y_i=Y_{i-1}$ and $X_i=X_{i-1}+v_i$, and therefore, $S^{*}_i=S^{*}_{i-1}+v_i$ and $h(Y_i)-h(Y_{i-1})=0$. 
If $v_i \in S^{*}_{i-1}$, then $h(S^{*}_{i-1})-h(S^{*}_i)=0$. If $v_i \notin S^{*}$, then by the submodularity,  
\begin{eqnarray*}
&&h(S^{*}_{i-1})-h(S^{*}_i) \\
&\leq& h(Y_{i-1}-v_i)-h(Y_{i-1})\\
&&\{\text{By Eq. (\ref{eq:accuracy})}\}\\
&\leq& \tilde{h}(Y_{i-1}-v_i)-\tilde{h}(Y_{i-1})+\frac{2\epsilon }{n}\cdot L^{*}= \tilde{b}_i \leq 0.
\end{eqnarray*}
Therefore, in both cases, we have $h(S^{*}_{i-1})-h(S^{*}_i)\leq 0\leq \frac{\tilde{a}_i}{2}$. Since $h(Y_i)=h(Y_{i-1})$,
\begin{eqnarray*}
&&h(S^{*}_{i-1})-h(S^{*}_i) \\
&\leq& \frac{1}{2}\tilde{a}+\frac{1}{2}(h(Y_i)-h(Y_{i-1}))\\
&=& \frac{1}{2}(\tilde{h}(X_{i-1}+v_i)-\tilde{h}(X_{i-1})+h(Y_i)-h(Y_{i-1}))\\
&&\{\text{By Eq. (\ref{eq:accuracy})}\}\\
&\leq& \frac{1}{2}(h(X_{i-1}+v_i)-h(X_{i-1})+h(Y_i)-h(Y_{i-1}))+\frac{\epsilon }{n} \cdot L^{*}\\
&\leq& \frac{1}{2}(h(X_{i})-h(X_{i-1})+h(Y_{i})-h(Y_{i-1}))+\frac{\epsilon }{n}\cdot h(S^{*})
\end{eqnarray*}
Thus, this case is proved.

\textbf{Case 2.} ($\tilde{a}_i < 0$ and $\tilde{b}_i \geq 0$). This case can be proved in a similar manner as that of Case 1.

\textbf{Case 3.} ($\tilde{a}_i \geq 0$ and $\tilde{b}_i > 0$). In this case, with probability $\tilde{a}_i/(\tilde{a}_i+\tilde{b}_i)$ (resp. $\tilde{b}_i/(\tilde{a}_i+\tilde{b}_i)$) that $X_i \leftarrow X_{i-1}+ v_i$ and $Y_i \leftarrow Y_{i-1}$ (resp. $X_i \leftarrow X_{i-1}$ and $Y_i \leftarrow Y_{i-1}\setminus v_i$) happen. 

Therefore, 

\begin{eqnarray}
\label{eq:lamma:appendix_2:lemma_2:eq_1}
&&\mathbb{E}[h(X_i)-h(X_{i-1})+h(Y_i)-h(Y_{i-1})] \nonumber\\
&=& \dfrac{\tilde{a}_i}{\tilde{a}_i+\tilde{b}_i}(h(X_{i-1}+v_i)-h(X_{i-1}))\nonumber\\
&&+ \dfrac{\tilde{b}_i}{\tilde{a}_i+\tilde{b}_i}(h(Y_{i-1}-v_i)-h(Y_{i-1}))\nonumber\\
&&\{\text{By Eq. (\ref{eq:accuracy})}\}\nonumber\\
&\geq& \dfrac{\tilde{a}_i}{\tilde{a}_i+\tilde{b}_i}(\tilde{h}(X_{i-1}+v_i)-\tilde{h}(X_{i-1})-\frac{2\epsilon }{n} \cdot L^{*})\nonumber\\
&&+ \dfrac{\tilde{b}_i}{\tilde{a}_i+\tilde{b}_i}(\tilde{h}(Y_{i-1}-v_i)-\tilde{h}(Y_{i-1})-\frac{2\epsilon }{n} \cdot L^{*})\nonumber\\
&=& \dfrac{\tilde{a}_i}{\tilde{a}_i+\tilde{b}_i}(\tilde{a}_i-\frac{4\epsilon }{n}\cdot L^{*})+ \dfrac{\tilde{b}_i}{\tilde{a}_i+\tilde{b}_i}(\tilde{b}_i-\frac{4\epsilon }{n}\cdot L^{*})\nonumber\\
&=& \dfrac{\tilde{a}_i^2+\tilde{b}_i^2}{\tilde{a}_i+\tilde{b}_i}-\frac{4\epsilon }{n}\cdot L^{*}\geq \dfrac{\tilde{a}_i^2+\tilde{b}_i^2}{\tilde{a}_i+\tilde{b}_i}-\frac{4\epsilon }{n}\cdot h(S^{*}).
\end{eqnarray}
Now let us consider $\mathbb{E}[h(S^{*}_{i-1})-h(S^{*}_i)]$. Note that  
\begin{eqnarray*}
&&\mathbb{E}[h(S^{*}_{i-1})-h(S^{*}_i)] \\
&=& \dfrac{\tilde{a}_i}{\tilde{a}_i+\tilde{b}_i}(h(S^{*}_{i-1})-h(S^{*}_{i-1}+v_i))\\
&&+ \dfrac{\tilde{b}_i}{\tilde{a}_i+\tilde{b}_i}(h(S^{*}_{i-1})-h(S^{*}_{i-1}-v_i))
\end{eqnarray*}
If $v_i \in S^{*}_{i-1}$ then $h(S^{*}_{i-1})-h(S^{*}_{i-1}+v_i)=0$ and
\begin{eqnarray*}
&&h(S^{*}_{i-1})-h(S^{*}_{i-1}-v_i)\\
&&\{\text{By submodularity}\}\\
&\leq& h(X_{i-1}+v_i)-h(X_{i-1})\\
&&\{\text{By Eq. (\ref{eq:accuracy})}\}\\
&\leq& \tilde{h}(X_{i-1}+v_i)-\tilde{h}(X_{i-1})+\frac{2\epsilon }{n}\cdot L^{*}= \tilde{a}_i
\end{eqnarray*} 
If $v_i \notin S^{*}_{i-1}$ then $h(S^{*}_{i-1}-v_i)-h(S^{*}_{i-1})=0$ and
\begin{eqnarray*}
&&h(S^{*}_{i-1})-h(S^{*}_{i-1}-v_i)\\
&&\{\text{By submodularity}\}\\
&\leq& {h}(Y_{i-1}-v_i)-{h}(Y_{i-1})\\
&&\{\text{By Eq. (\ref{eq:accuracy})}\}\\
&\leq& \tilde{h}(Y_{i-1}-v_i)-\tilde{h}(Y_{i-1})+\frac{2\epsilon }{n}\cdot L^{*}= \tilde{b}_i.
\end{eqnarray*} 
Therefore, in either case, $\mathbb{E}[h(S^{*}_{i-1})-h(S^{*}_i)] \leq \dfrac{\tilde{a}_i\tilde{b}_i}{\tilde{a}_i+\tilde{b}_i}$. Together with Eq. (\ref{eq:lamma:appendix_2:lemma_2:eq_1}), $\mathbb{E}[h(S^{*}_{i-1})-h(S^{*}_i)] \leq \frac{1}{2}\mathbb{E}[h(X_i)-h(X_{i-1})+h(Y_i)-h(Y_{i-1})]+\frac{2\epsilon }{n}\cdot h(S^{*})$
\end{proof}
Now new are ready to prove Theorem \ref{theorem:accuracy}, as shown in the following lemma.
\begin{lemma}
$\mathbb{E}[h(X_n)]\geq (\frac{1}{2}-\epsilon)\cdot h(S^{*})$.
\end{lemma}
\begin{proof}
Summing the inequality of Lemma \ref{lamma:appendix_2:lemma_2} for $i$ from 1 to $n$ yields that $\mathbb{E}[h(S^{*}_{0})-h(S^{*}_n)] \leq \frac{1}{2}\mathbb{E}[h(X_n)+h(Y_n)-h(X_0)-h(Y_0)]+2\epsilon \cdot h(S^{*})$. Since $S^{*}_0=S^{*}$ and $S^{*}_n=X_n=Y_n$, we have $\mathbb{E}[h(X_n)]\geq(\frac{1}{2}-\epsilon) \cdot h(S^{*})$.
\end{proof}

\subsection{Proof of Lemma \ref{lemma:ra_mean}}
\label{sec:proof_lemma:ra_mean}
Let $\pi(S,v)$ be the probability that $v$ can be an adopter under $S$ and therefore $\pi(S)=\sum_{v \in V} \pi(S,v)$ due to the linearity of expectation. Let $\mathcal{G}_v^S$ be the set of the realizations where $v$ is an adopter under $S$ and $\Pr[g]$ be the probability that $g$ can be sampled according the triggering distribution.  By Lemma \ref{lemma:equivalent}, $\pi(S,v)=\sum_{g \in \mathcal{G}_v^S}\Pr[g]$. On the other hand, let $x(S,R|v)$ be the conditional variable of $x(S,R)$ when $v$ is selected in line 2 of Alg. \ref{alg:ra_set}. Thus, $\mathbb{E}[x(S,R)]=\frac{\sum{v\in V}\mathbb{E}[x(S,R|v)]}{n}$. Note that $v$ is an adopter under $S$ in $g$ if and only if at least one node in $S$ in reverse-adopted-reachable to $v$. Therefore, $\mathbb{E}[x(S,R|v)]=\sum_{g \in \mathcal{G}_v^S}\Pr[g]=\pi(S,v)$, which completes the proof.

\subsection{Proof of Lemma \ref{lemma:epsilon_1}}
\label{sec:proof_lemma:epsilon_1}
By rearrangement, $\Pr[F(\mathcal{R}_l,S)-f(S)>\epsilon_1 f(V_{opt})]$ is equal to $\Pr[ {\sum_{i=1}^{l}x(R_i,S)}-\frac{l\pi(S)}{n}>l \cdot \frac{\pi(S)}{n} \frac{\epsilon_1 \cdot f(V_{opt})}{P\cdot \pi(S)}]$. According to Chernoff bound, this probability is no larger than $$\exp(-\frac{l \cdot \frac{\pi(S)}{n} \cdot (\frac{\epsilon_1 \cdot f(V_{opt})}{P\cdot \pi(S)})^2}{2+\frac{\epsilon_1 \cdot f(V_{opt})}{P\cdot \pi(S)}}).$$ By Eq. (\ref{eq:lower_bound}) and $\pi(S)\leq n$, it is further no larger than $\exp(-\frac{l \cdot \epsilon_1^2 \cdot (P-C)^2}{2 P^2+\epsilon_1 P \cdot(P-C)})$. Finally, by Eq.  (\ref{eq:delta_1}), this probability is at most $\frac{1}{N \cdot  2^n}$.
Since there are at most $2^n$ subsets, by the union bound, the lemma holds.

\subsection{Proof of Lemma \ref{lemma:epsilon_2}}
\label{sec:proof_lemma:epsilon_2}
By rearrangement, $\Pr[F(\mathcal{R}_l,V_{opt})-f(V_{opt})\leq -\epsilon_2 f(V_{opt})]$ is equal to $\Pr[ {\sum_{i=1}^{l}x(R_i,V_{opt})}-l\cdot\frac{ \pi(V_{opt})}{n} \leq -l \cdot \frac{\pi(V_{opt})}{n} \frac{\epsilon_2 \cdot f(V_{opt})}{P\cdot \pi(V_{opt})}]$. By Chernoff bound, this probability is no larger than $$\exp(-\frac{l \cdot \frac{\pi(V_{opt})}{n} \cdot (\frac{\epsilon_2 \cdot f(V_{opt})}{P\cdot \pi(V_{opt})})^2}{2}).$$ By Eq. (\ref{eq:lower_bound}) and $\pi(S)\leq n$, it is then no larger than $\exp(-\frac{l \cdot \epsilon_2^2 \cdot (P-C)^2}{2P^2})$, which, by Eq. (\ref{eq:delta_2}), is equal to or smaller than $\frac{1}{N}$.

\subsection{Proof of Lemma \ref{lemma:ra_s_epsilon_2}}
\label{sec:proof_lemma:ra_s_epsilon_2}
This proof of this lemma is exactly the same as that of Lemma \ref{lemma:epsilon_2} because $\delta_2^*=\delta_2$ and the lower tail inequality of Chernoff bound (i.e. Eq. (\ref{eq:chernoff_2})) is identical to that of the Martingale (i.e. Eq. (\ref{eq:maringale_2})).
\subsection{Proof of Lemma \ref{lemma:ra_s_epsilon_1}}
\label{sec:proof_lemma:ra_s_epsilon_1}
\begin{eqnarray*}
&&　\Pr[\widetilde{f}_l(V^{*})-f(V^{*}) > \epsilon_1 \cdot f(V_{opt})] \\
&=&　\Pr[P\cdot \widetilde{\pi}_l(V^{*})-P\cdot \pi(V^{*}) > \epsilon_1 \cdot f(V_{opt})] \\
&=&　\Pr[|l\cdot \frac{\widetilde{\pi}_l(V^{*})}{n}-l\cdot \frac{\pi(V^{*})}{n}| > l \cdot  \frac{\pi(V^{*})}{n} \frac{\epsilon_1 f(V_{opt})}{P\cdot \pi(V^{*})}] \\
&&\{\text{Since~} f(V_{opt}) \geq n(P-C) \}\\
&\leq&　\Pr[|l\cdot \frac{\widetilde{\pi}_l(V^{*})}{n}-l\cdot \frac{\pi(V^{*})}{n}| > l \cdot  \frac{\pi(V^{*})}{n} \frac{\epsilon_1 n(P-C)}{P\cdot \pi(V^{*})}] 
\end{eqnarray*}
Applying the Chernoff bound, the above probability is no larger than $\exp(-\frac{l \cdot \frac{\pi(V^{*})}{n} \cdot \frac{\epsilon_1^2 n^2 r^2}{\pi^2(V^{*})}}{2+\frac{\epsilon_1 nr}{\pi(V^{*})}})$. Finally, by Eq. (\ref{eq:delta_3}) and $\pi(V^{*})\leq n$, it is no larger than $1/N$.

\subsection{Proof of Lemma \ref{lemma:ra_s_case_b}}
\label{sec:proof_lemma:ra_s_case_b}
We first prove a result which is analogous to Lemma \ref{lemma:epsilon_1}. Let $\mathcal{R}_l=\{R_1...,R_l\}$ be the RA sets used in line 7 of Alg. \ref{alg:reverse_2} and $V^{*}$ be the final result returned in line 8, where $l \geq \delta_1^*$. For any seed set $S$, by rearrangement, $\Pr[F(\mathcal{R}_l,S)-f(S)>\epsilon_1 f(V_{opt})]$ is equal to 
$$\Pr[ {\sum_{i=1}^{l}x(R_i,S)}-l \cdot \frac{\pi(S)}{n}>l \cdot \frac{\pi(S)}{n} \frac{\epsilon_1 \cdot f(V_{opt})}{P\cdot \pi(S)}].$$
Applying Eq. (\ref{eq:maringale_2}) yields that the above probability is no larger than 
$$\exp(-\frac{l \cdot \frac{\pi(S)}{n} \cdot (\frac{\epsilon_1 \cdot f(V_{opt})}{P\cdot \pi(S)})^2}{2+\frac{2\epsilon_1 \cdot f(V_{opt})}{3P\cdot \pi(S)}}).$$
By Eq. (\ref{eq:lower_bound}) and $\pi(S)\leq n$, this probability is no larger than 
$$\exp(-\frac{3l \cdot \epsilon_1^2 \cdot (P-C)^2}{6 P^2+2\epsilon_1 P \cdot(P-C)}).$$
Finally, replacing $l$ with $\delta_1^*$ implies that 
$$\exp(-\frac{3l \cdot \epsilon_1^2 \cdot (P-C)^2}{6 P^2+2\epsilon_1 P \cdot(P-C)}) \leq \frac{1}{N\cdot 2^n}.$$
Since there are at most $2^n$ subsets, with probability at most $1/N$, there exists some $S \subseteq V$ such that $F(\mathcal{R}_l,S)-f(S) > \epsilon_1 \cdot f(V_{opt})$, which means, 
$F(\mathcal{R}_l,V^{*})-f(V^{*}) \leq \epsilon_1 \cdot f(V_{opt})$ holds with the probability at least $1-1/N$. Therefore,
\begin{eqnarray*}
f(V^{*}) &\geq& F(\mathcal{R}_l,V^{*})- \epsilon_1 \cdot f(V_{opt})\\
&&\{\text{By Lemma \ref{lemma:1/2}}\}\\
&\geq& \frac{1}{2}F(\mathcal{R}_l,V_{opt})- \epsilon_1 \cdot f(V_{opt})\\
&&\{\text{By Lemma \ref{lemma:ra_s_epsilon_2}}\}\\
&\geq& \frac{1}{2}(1-\epsilon_2)f(V_{opt})- \epsilon_1 \cdot f(V_{opt})\\
&&\{\text{By Eqs \ref{eqs:1}}\}\\
&\geq& (\frac{1}{2}-\epsilon) \cdot f(V_{opt})
\end{eqnarray*}
holds with probability at least $1-2/N$.

\vfill


\end{document}